\documentclass[12pt,draftcls,onecolumn]{IEEEtran}

\usepackage{subfigure}
\usepackage{amssymb, amsmath, amsfonts, mathtools}
\usepackage{empheq}
\usepackage{caption}
\usepackage{cite}
\usepackage{bm}
\usepackage{optidef}
\usepackage{bbm}
\usepackage{dsfont}
\usepackage{amsthm}
\usepackage{tikz}
\usepackage{algorithm}
\usetikzlibrary{external}
\usepackage{pgfplots}
\usepackage{cases}

\newtheorem{theorem}{Theorem}
\newtheorem{definition}{Definition}
\newtheorem{lemma}{Lemma}

\newtheorem{remark}{Remark}
\newtheorem{assumption}{Assumption}

\newlength\figureheight
\newlength\figurewidth

\ifnum\pdfshellescape=1
\fi

\newcommand{\B}{{\mathbb{B}}}

\newcommand{\R}{{\mathbb{R}}}

\newcommand{\I}{{\mathcal{I}}}
\newcommand{\M}{{\mathcal{M}}}

\DeclareMathOperator*{\argmin}{arg\; min}     

\DeclareMathOperator{\supp}{supp}
\DeclareMathOperator{\rad}{rad}
\DeclareMathOperator{\dist}{dist}
\DeclareMathOperator{\avg}{avg}

\DeclareMathOperator{\diag}{diag}
\DeclareMathOperator{\var}{var}
\DeclareMathOperator{\res}{res}
\DeclareMathOperator{\chv}{chv}
\DeclareMathOperator{\minimize}{minimize}
\DeclareMathOperator{\subjecto}{subject\,to}

\DeclareMathOperator{\trm}{trm}
\DeclareMathOperator{\mm}{mm}
\DeclareMathOperator{\mv}{mv}
\DeclareMathOperator{\id}{id}

\newcommand\addtag{\refstepcounter{equation}\tag{\theequation}}

\makeatletter

\newcommand{\Rmnum}[1]{\expandafter\@slowromancap\romannumeral #1@}
\makeatother

\DeclareFontFamily{OT1}{pzc}{}
\DeclareFontShape{OT1}{pzc}{m}{it}{<-> s * [1.000] pzcmi7t}{}
\DeclareMathAlphabet{\mathpzc}{OT1}{pzc}{m}{it}

\title{Secure State Estimation with Byzantine Sensors: A Probabilistic Approach}

\author{Xiaoqiang Ren$^1$, Yilin Mo$^{2*}$, Jie Chen$^{3}$, and Karl H. Johansson$^{1}$
\thanks{$1$:
School of Electrical Engineering and Computer Science, KTH Royal Institute of Technology, 114 28 Stockholm, Sweden, (Emails: xiaren, kallej@kth.se)}%
\thanks{$2$: Department of Automation, Tsinghua University, Beijing, China,
	(Email: ylmo@tsinghua.edu.cn)}
\thanks{$3$: Department of Electronic Engineering, City University of Hong Kong, Hong Kong, China, (Email: jichen@cityu.edu.hk)}
  \thanks{$*$: Corresponding Author.}
}

\begin{document}
\maketitle
\begin{abstract}
  This paper studies static state estimation in multi-sensor settings, with a caveat that an unknown subset of the sensors are compromised by an adversary, whose measurements can be manipulated arbitrarily. The attacker is able to compromise $q$ out of $m$ sensors. 
  A new performance metric, which quantifies the asymptotic decay rate for the probability of having an estimation error larger than $\delta$, is proposed. We develop an optimal estimator for the new performance metric with a fixed $\delta$, which is the Chebyshev center of a union of ellipsoids. 
  We further provide an estimator that is optimal for every $\delta$, for the special case where the sensors are homogeneous. 
  Numerical examples are given to elaborate the results.

\end{abstract}

\section{Introduction}
In cyber-physical systems, numerous sensors with limited capacity are spatially deployed and connected via ubiquitous wired and wireless communication networks. This makes it nearly impossible to guarantee the security of every single sensor or communication channel. Therefore, security problems of cyber-physical systems have attracted much attention recently, e.g.,~\cite{mo2012cyber,teixeira2015secure}.

Robust estimation has been studied over decades to deal with uncertainties of input data~\cite{hampel1974influence, kassam1985robust, huber2011robust}. The robustness is usually measured  by  influence functions or breakdown point, and several celebrated estimators have been developed, such as M-, L-, and R-estimators. The limitation of this robustness theory is the assumption that the bad data are independent~\cite{huber2011robust}, which, however, is not the case in general for cyber attacks. The fact that compromised sensors may cooperate and the estimation is done sequentially makes the ``bad'' data correlated both spatially and temporarily. 

Recently, dynamic state estimation with some Byzantine sensors has been discussed. Most approaches in the existing literature can be classified into two categories: stacked measurements~\cite{fawzi2014secure, pajic2017attack, mishra2017secure} and Kalman filter decomposition~\cite{mo2016secure, liu2017secure}. Fawzi \textit{et al.}~\cite{fawzi2014secure} used the stacked measurements from time $k$ to $k+T-1$ to estimate the state at time $k$ and provided $l_0$ and $l_1$-based state estimation procedures. Since deterministic systems are concerned, the $l_0$-based  procedure can exactly recover the state.  Pajic \textit{et al.}~\cite{pajic2017attack} extended the deterministic systems in~\cite{fawzi2014secure} to ones with bounded measurement noises and obtained upper bounds of estimation error for both $l_0$ and $l_1$-based estimators. Mishar \textit{et al.}~\cite{mishra2017secure}  studied stochastic systems with unbounded noises and proposed a notion of $\epsilon$-effective attack. The state estimation there is in essence an attack detection problem;  a Chi-squared test is applied to the residues and the standard Kalman filter output based on the measurements from the largest set of sensors that are deemed $\epsilon$-effective attack-free is used as the state estimate. Notice that to detect the $\epsilon$-effective attack-free sensors correctly with high probability, the window size $T$ must be large enough. 
The authors did not provide estimators before detection decisions are made. 
The authors of~\cite{mo2016secure, liu2017secure} used local estimators at each sensor and proposed a LASSO based fusion scheme. However, their approach imposes some strong constraints on the system dynamics. Furthermore, the estimate error of the proposed algorithm when there are indeed attacks is not specifically characterized.

In this paper, we deal with scenarios where noises are not necessarily bounded and give a different characterization of the estimator performance, i.e., the decaying rate of the worst-case probability that the estimation error is larger than some value $\delta$ rather than the worst-case error in~\cite{pajic2017attack, mo2016secure, liu2017secure} and estimation error covariance in~\cite{mishra2017secure}. This is partially motivated by the following three observations. Firstly, with unbounded noise, the worst-case estimation error might result in too conservative system designs.
Notice also that even for the bounded noise cases studied in~\cite{pajic2017attack}, the upper bound of the worst-case estimation error thereof increases with respect to (w.r.t.) the window size $T$, which counters intuition since more information should lead to better estimation accuracy. Secondly, to mitigate the bad effects caused by Byzantine sensors, one has to accumulate much enough information, i.e., the time window $T$  should be large enough. In this case, the decaying rate is able to characterize the probability well enough (just as, e.g.,~\cite{mishra2017secure}). Lastly, the system operator may pre-define the error threshold $\delta$ according to the performance specification, which leads to a more flexible system design.   

%

In the subsequent sections, we focus on the problem of secure static state estimation with Byzantine sensors.
A fusion center aims to estimate a vector state $x\in\R^n$ from measurements collected by $m$ sensors, among which $q$ sensors might be compromised. Without imposing any restrictions on the attacker's capabilities, we assume that the compromised sensors can send arbitrary messages.  Static state estimation has a wide range of applications in power system, where the power network states (i.e., bus voltage phase angles and bus voltage magnitudes) are estimated from measurements collected by Supervisory Control And Data Acquisition (SCADA) systems (e.g., transmission line power flows, bus power injections,
and part of the bus voltages) through remote terminal units (RTUs)~\cite{schweppe1974static, abur2004power}.  Considering the possibility that the RTUs are controlled and the communicated data from SCADA systems tampered with by malicious attackers, much work has devoted to security problems of power systems, e.g.~\cite{liu2011false, mo2011cyber, hendrickx2014efficient, sou2014data}.
The closest literature is~\cite{mo2015secure,HanMX15}, which, however, both focused on the one-shot scenario, while in this work the observations are taken sequentially, the possible temporal correlations of which make the analysis more challenging. 
We should also note that both~\cite{mo2015secure,HanMX15} used the worst-case estimate error as the performance metric rather than the probabilistic approach in this paper. Moreover, the main results of this work provide fundamental insights on the counterpart for dynamical systems that we are still investigating.


The main contributions of this work are summarized as follows.
\begin{enumerate}
  \item We propose a new metric to characterize the performance of an estimator when observation noise is not necessarily bounded and an attacker may be present. 
  \item We provide an optimal estimator for a given estimation error threshold $\delta$ (Theorem~\ref{theorem:optimalestimator}), which is the Chebyshev center of a union of ellipsoids. We then propose an  algorithm to compute the optimal estimator (Algorithm~\ref{alg:fstar} and Theorem~\ref{theorem:performanceguaAlg}). 
  \item When the sensors are homogeneous, we further provide a uniformly optimal estimator, i.e., simultaneously optimal for any error threshold $\delta$ (Theorem~\ref{theorem:uniform}). The estimator is just the ``trimmed mean" of the averaged observations.
\end{enumerate}

A preliminary version of this paper was presented in~\cite{ren2018secure}. The main difference is threefold. Firstly, new results have been provided in this paper, i.e., numerical implementation of our algorithm (Section~\ref{sec:numericalImple}) and uniformly optimal estimator design (Section~\ref{sec:uniformoptimal}). Secondly, in~\cite{ren2018secure}, only proofs of Lemmas~\ref{lemma:probgeometric}~and~\ref{lemma:probinter} were presented due to page limitation. Lastly, new simulations have been conducted in this paper for better illustration.


\emph{Organization:} In Section~\ref{sec:problem}, we formulate the problem of static state estimation with Byzantine sensors, including the attack model and performance metric.
The main results are presented in Section~\ref{sec:secure-estimator}.
We first prove that one may only consider estimators with certain ``nice" structures. 
Based on this, we then provide an optimal estimator for a given error threshold and propose an algorithm to compute the optimal estimator.
Furthermore, a very simple yet uniformly optimal estimator when sensors are homogeneous is provided in Section~\ref{sec:uniformoptimal}.
After showing numerical examples in Section~\ref{sec:simulation}, we conclude the paper in Section~\ref{sec:conclusion}.
All proofs are reported in the appendix.

\emph{Notations}: $\R$ ($\mathbb{R}_{+}$) is the set of (nonnegative) real numbers. $\mathbb{N}$ ($\mathbb{N}_{+}$) is the set of nonnegative (positive) integers. For a vector $x\in\R^n$, define $\|x\|_0$ as the ``zero norm'', i.e., the number of nonzero elements of the vector $x$.  
For a vector $x\in\R^n$, the support of $x$, denoted by $\supp(x)$, is the set of indices of nonzero elements:
\[\supp( x) \triangleq \{i\in\{1,2,\ldots,n\}:  x_i \neq 0\}.\] 
Define ${\bm 1}$ as the column vector of ones and the size clear from the context if without further notice. Let  ${\bm I}_n$ be the identity matrix of size $n\times n$.
For a matrix ${\bm M}\in\R^{m\times n}$, unless stated otherwise, ${\bm M}_i$ represents the $i$-th row, and  ${\bm M}_{\I}$  the matrix obtained from ${\bm M}$ after removing all of the rows except those in the index set $\I$. 
For a set of matrices $\mathcal{A}\subseteq \R^{m\times n}$, we use $\mathcal{A}_{\mathcal I}$ to denote the set of matrices keeping rows indexed by $\mathcal I$, i.e.,
\[\mathcal{A}_{\mathcal I} \triangleq \{{\bm M}_{\I}: {\bm M}\in\mathcal{A}\}.\]
For a set $\mathcal{A}$, define the indicator function as $\mathbbm{1}_{\mathcal{A}}(x) = 1$, if $x\in\mathcal{A}$; $0$ otherwise.
The cardinality of a set $\mathcal A$ is denoted as $|\mathcal A|$. Let ${\bm M}^{\top}$ denote the transpose of the matrix ${\bm M}$. We write ${\bm M} \succcurlyeq {\bm N}$ if ${\bm M} - {\bm N}$ is a positive semi-definite matrix.

\section{Problem Formulation}
\label{sec:problem}
\subsection{System Model}  \label{sec:StaticEstimation}
Consider the problem of estimating the state $x\in\R^n$ using $m$ sensor measurements as depicted in Fig.~\ref{Fig:blockdiagram}. Let $\M \triangleq \{1,\ldots,m\}$ be the index set of all the sensors. The measurement equation for sensor $i\in\M$ is
\begin{align*}
 z_i(k) = H_ix + w_i(k),
\end{align*}
where $z_i(k)\in\R$ is the 
(``true'') measurement collected by the sensor $i$ at time $k\in\mathbb N_+$,
$H_i\in\R^{1\times n}$ is the output matrix associated with sensor $i$, $w_i(k)\in\R$ is the observation noise. It is assumed that $w_i(k)$ is Gaussian distributed with zero mean and variance $\mathbb E[(w_i(k))^2] = W_i>0$ for any $i,k$\footnote{Actually, the main results  in this paper hold for any noise distribution in the exponential family; the details are discussed in Remark~\ref{remark:exponentialfamily}}. Furthermore, $w_i(k)$ are independent across the sensors and over time, i.e., $\mathbb E[w_{i_1}(k_1)w_{i_2}(k_2)] = 0$ if $i_1\neq i_2$ or $k_1\neq k_2$.

In the presence of attacks, the measurement received by the fusion center is $y_i(k)$, with satisfies the following equation:
\begin{align*}
y_i(k) = z_i(k) + a_i(k),
\end{align*}
where $a_i(k)\in \R$ is the bias injected by the attacker.
\begin{center}
	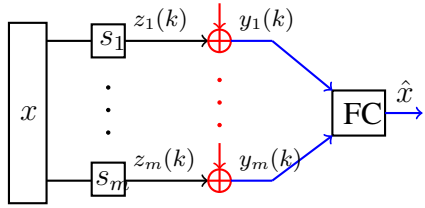
\begin{figure}[ht]
		\centering
		\begin{tikzpicture}
\draw [thick] (0.1,0.1) rectangle (0.54,0.54);
\node [right] at (0,0.3) {$s_m$};

\draw [fill] (0.32,0.94) circle [radius=0.022];
\draw [fill] (0.32,1.24) circle [radius=0.022];
\draw [fill] (0.32,1.54) circle [radius=0.022];

\draw [thick] (0.1,1.94) rectangle (0.54,2.38);
\node [right] at (0.05,2.16) {$s_1$};

\draw [->,thick] (0.54,2.16) -- (1.65,2.16);
\draw [thick, red] (1.8,2.16) circle [radius=0.15];
\draw [-,thick,red] (1.65,2.16) -- (1.95,2.16);
\draw [-,thick,red] (1.8,2.01) -- (1.8,2.31);
\draw [->,thick,red] (1.8,2.66) -- (1.8,2.31);

\draw [->,thick] (0.54,0.3) -- (1.65,0.3);
\draw [thick, red] (1.8,0.3) circle [radius=0.15];
\draw [-,thick,red] (1.65,0.3) -- (1.95,0.3);
\draw [-,thick,red] (1.8,0.15) -- (1.8,0.45);
\draw [->,thick,red] (1.8,0.8) -- (1.8,0.45);

\draw [fill, red] (1.8,1.04) circle [radius=0.022];
\draw [fill, red] (1.8,1.34) circle [radius=0.022];
\draw [fill, red] (1.8,1.64) circle [radius=0.022];

\draw [thick] (-1,0) rectangle (-0.5,2.4);
\node [right] at (-1,1.2) {$x$};
\draw [-,thick] (-0.5,0.3) -- (0.1,0.3);
\draw [-,thick] (-0.5,2.16) -- (0.1,2.16);

\draw [-,thick, blue] (1.95,0.3) -- (2.5,0.3);
\draw [-,thick, blue] (1.95,2.16) -- (2.5,2.16);

\draw [thick] (3.3,0.9) rectangle (4,1.5);
\node [right] at (3.3,1.2) {FC};
\draw [->,thick, blue] (2.5,2.16) -- (3.3,1.5);
\draw [->,thick, blue] (2.5,0.3) -- (3.3,0.9);
\draw [->,thick, blue] (4,1.2) -- (4.5,1.2);
\node [above right] at (4,1.2) {$\hat{x}$};

{\footnotesize
\node [above right] at (0.54,0.3) {$z_m(k)$};
\node [above right] at (0.54,2.16) {$z_1(k)$};
\node [above right] at (1.95,0.3) {$ y_m(k)$};
\node [above right] at (1.95,2.16) {$ y_1(k)$};
}

\end{tikzpicture}
		\caption{The fusion center (FC) estimates the underlying state $x$ using sensor measurements  that might be manipulated.}
		\label{Fig:blockdiagram}
	\end{figure}
\end{center}
 We assume the attacks are $q$-sparse:
\begin{assumption}[$q$-sparse attack] \label{assumpt:sparseattack}
There exists an index set $\mathcal{C} \subseteq \M$ such that
\begin{enumerate}
  \item for any sensor $i\in\M\setminus\mathcal C$, $a_i(k)=0$ for any time $k$.
  \item $|\mathcal C|=q$.
\end{enumerate}
\end{assumption}
The sparse attack model, which is conventional in the literature~\cite{fawzi2014secure, pajic2017attack, mishra2017secure,  mo2016secure, liu2017secure, fellouris2017efficient,  xiaoqiang_securedetection, HanMX15}, says that the set of compromised sensors is somewhat ``constant'' over time.   This
is in essence the only restriction we impose on the attacker's capability. The bias $a_i(k)$ of a compromised sensor may take any value and might be correlated across sensors and over time. If the set of compromised sensors is time-varying, the estimators (or detectors) in \emph{all} the aforementioned literature will be destroyed. That is, the estimators (or detectors) could not work at all or the error could be arbitrarily large. In this paper, without this constant property, even Lemma~\ref{lemma:compressed_estimator} provided later (in particular, e.g.,~\eqref{eqn:optCompress}~and~\eqref{eqn:optCompress1}), which is the basis for Theorems~\ref{theorem:restrict_estimator}~and~\ref{theorem:optimalestimator}, would not hold.
\begin{assumption}[System  knowledge] \label{assumpt:knowledge2}
	The system designer knows the number $q$, but does not know the exact set of compromised sensors $\mathcal C$. 
\end{assumption}
The quantity $q$ might be determined by the \textit{a priori} knowledge about the quality of each sensor.
Alternatively, the quantity $q$ may be viewed as a design parameter, which indicates the resilience level  that the designer is willing to pay for.  One finds more comments about the above assumption in Remark~\ref{remark:parameterq}.

Let ${\bm H}= [H_1^\top, H_2^\top, \ldots, H_m^\top]^\top$ be the measurement matrix. 
We assume that the matrix ${\bm H}$ is $2q$-observable:
\begin{assumption} \label{assumpt:2qobservable}
	The measurement matrix ${\bm H}$ is $2q$-observable, i.e., for every set $\I\subseteq\M$ with $|\I|=m-2q$, the matrix ${\bm H}_{\I}$ is of full column  rank.
\end{assumption}
It has been shown in~\cite{fawzi2014secure} that  $2q$-observability of the measurement matrix is a necessary and sufficient condition to recover the exact state under $q$-sparse attacks when there are no observation noises. One finds the results if Assumption~\ref{assumpt:2qobservable} is violated in Lemma~\ref{lemma:vioassumpt2} later.  Notice that in power systems, measurement redundancy is a common practice~\cite{abur2004power}. 

To introduce the knowledge available at the attacker, we need the following definitions. 
Define the measurement from all sensors at time $k$ to be a column vector:
\begin{equation}
{\bm y}(k)\triangleq \begin{bmatrix}
y_1(k)&y_2(k)&\ldots&y_m(k)
\end{bmatrix}^\top\in \mathbb R^m.
\label{eq:yvector}
\end{equation}
We further define ${\bm Y}(k)$ as a matrix of all measurements from time $1$ to time $k$:
\begin{equation}
{\bm Y}(k) \triangleq \begin{bmatrix}
{\bm y}(1)&{\bm y}(2)&\ldots&{\bm y}(k)
\end{bmatrix}\in \mathbb R^{m\times k}.
\label{eqn:bigY}
\end{equation}
The quantities ${\bm a}(k), {\bm A}(k)$ are defined in the same manner.  At time $k$, given measurements from all the sensors ${\bm Y}(k)$, the fusion center generates a state estimate $\hat{x}_k$. The estimator $f$ might be random, i.e., given ${\bm Y}(k)$, $\hat{x}_k$ is a random variable governed by certain probability measure on $\R^n$ determined by~$f$.   

\begin{assumption}[Attacker's knowledge] \label{assumpt:knowledge}
	It is assumed that
	\begin{enumerate}
		\item the attacker knows the true state $x$;
		\item the attacker knows the estimator $f$, the system parameters (i.e., each $H_i$ and $W_i$), and can access the historical and current observations from the compromised sensors.
	\end{enumerate}
\end{assumption}
The above assumption as a whole has been adopted in literature on sparse attack, e.g.,~\cite{xiaoqiang_securedetection,fellouris2017efficient, marano2008distributed}, while the second bullet prevails in literature on data-injection attack, e.g.,~\cite{liu2011false,mishra2017secure, smith2011decoupled}. The parameters $H_i$ and $W_i$ might be developed by an attacker using the \textit{a priori} knowledge of the underlying physical model.  To obtain the true state, the attacker may deploy its own sensor network. Though it might be difficult in practice to obtain the accurate parameters and true state for an attacker, this assumption is de facto when dealing with potential worst-case attacks. We should note that this assumption is in accordance with the Kerckhoffs's principle~\cite{shannon1949communication}, namely the security of a system should not rely on its obscurity. Interested readers are referred to~\cite{teixeira2012attack}  to see more attack models in cyber-physical systems. This assumption is leveraged later to define the performance metric in~\eqref{eqn:errorProb} and characterize the attack capacity in Theorems~\ref{theorem:restrict_estimator}~and~\ref{theorem:optimalestimator}. In particular, one finds more on how Assumptions~\ref{assumpt:sparseattack}~and~\ref{assumpt:knowledge} are utilized to derive~\eqref{eqn:optCompress} in Remark~\ref{remark:relation} later.
\subsection{Performance Metric}
At time $k$, given the measurements ${\bm Y}(k)_{\mathcal{C}}$, the bias ${\bm A}(k-1)$, the set of compromised sensors $\mathcal C$,  and true state $x$, the bias ${\bm a}(k)$ is generated according to some probability measure on $\R^m$. This bias injection mechanism is denoted by $g$. Let $\mathcal G$ be the set of all attack strategies such that the generated bias ${\bm a}(k)$ satisfies the $q$-sparse attack model in Assumption~\ref{assumpt:sparseattack}.

In this paper, we are concerned with the worst-case scenario. Given an estimator $f$, we define
\begin{equation}  \label{eqn:errorProb}
  e(f,k,\delta) \triangleq \sup_{ \mathcal{C}\subseteq\mathcal{M},g\in\mathcal{G}, x\in\R^n} \mathbb{P}_{f,g,x,\mathcal{C}}\left(  \|\hat x_k -x \|_2 > \delta    \right)
\end{equation}
as the worst-case probability that the distance between the estimate at time $k$ and the true state is larger than a certain value $\delta \in\R_+$ considering all possible attack strategies, the set of compromised sensors and the true state. We use $\mathbb{P}_{f,g,x,\mathcal{C}}$ to denote the probability measure governing $\hat{x}_k$ when the estimator $f$, attack strategy $g$, the true state $x$, and the set of compromised sensors $\mathcal{C}$ are given.

Ideally, one wants to design an estimator $f$ such that $e(f,k,\delta)$ is minimized at any time $k$ for any $\delta$. However, it is quite difficult to analyze $e(f,k,\delta)$ when $k$ takes finite values since computing the probability of error usually involves numerical integration.
Therefore,  we consider an asymptotic estimation performance, i.e., the exponential rate with which the worst-case probability  goes to zero:
\begin{equation} \label{eqn:RateObj}
 r(f,\delta) \triangleq \liminf_{k\rightarrow\infty} - \frac{\log e(f,k,\delta)}{k}.
\end{equation}
 Obviously, for any $\delta$, the system designer would like to maximize $r(f,\delta)$ by choosing a suitable estimator $f$. 

The threshold $\delta$ is chosen by the designer in accordance with system accuracy
 requirement by noticing that a true state $x$ is perceived as the same with any point $x'$ lying inside its neighbourhood, i.e., $ \|x' -x \|_2 \leq \delta$ by the above performance metric. However, in some cases (see Section~\ref{sec:uniformoptimal}), there is no need to determine $\delta$ since one can find an estimator that simultaneously maximizes $r(f,\delta)$ for all $\delta$.

\subsection{Problems of Interest} \label{sec:problemofintrest}
The following three problems are to be addressed.
\begin{enumerate}
  \item  \emph{Performance limit.} For any $\delta$, what is the maximal rate $r(f,\delta)$ that can be achieved by all possible estimators?
  \item  \emph{Optimal estimator.}  Given $\delta$, what is the optimal estimator that maximizes  $r(f,\delta)$?
  \item  \emph{Uniform optimality.} Is there an estimator that simultaneously maximizes $r(f,\delta)$ for all $\delta >0$?
  
\end{enumerate}

\section{Optimal Estimator}
\label{sec:secure-estimator}
In this section, the first two problems in Section~\ref{sec:problemofintrest} shall be addressed. We provide an estimator based on Chebyshev centers,  prove its optimality, and further present a numerical algorithm to implement it.

\subsection{Compressed and Deterministic Estimator}
A generic estimator $f_k$ might \emph{randomly} generate an estimate $\hat{x}_k$ based on \emph{all} the information contained in ${\bm Y}(k)$. In other words, given ${\bm Y}(k)$, the estimate $\hat x_k$ might be a random variable; and if any element (totally there are $m\times k$) of two observation matrices, say ${\bm Y}(k)$ and ${\bm Y}'(k)$, is different,  the corresponding probability distributions of the estimate $\hat x_k$ might be different. In this subsection, however, we shall show that, without loss of optimality, one may only consider estimators with certain ``nice'' structure (i.e., the compressed and deterministic estimators defined in Definition~\ref{def:deterministicandcompressed} later).


Define an operator $\avg(\cdot)$ that averages each row of the inputed real-valued matrix, i.e., for any matrix ${\bm M}\in\R^{n_1\times n_2}$, 
\[\avg({\bm M}) \triangleq {\bm M}{\bm 1}/n_2.\]
Hence, $\avg({\bm Y}(k))$ is a vector in $\R^m$ and the $i$-th element is the empirical mean of the observation from time $1$ to $k$ available for sensor $i$.

We use $\mathbb{P}_f(\hat{x}_k|{\bm Y}(k))$ to denote the conditional probability measure of estimate  $\hat{x}_k$ given any estimator $f$ and the information ${\bm Y}(k)$. 
Notice that an estimator $f$ can be completely characterized by the sequence of conditional probability measures from time $1$ to $\infty$: $(\mathbb{P}_f(\hat{x}_1|{\bm Y}(1)), \mathbb{P}_f(\hat{x}_2|{\bm Y}(2)),\ldots   )$. 
\begin{definition}
	An estimator $f$ is said to be compressed if at each time $k$, it only utilizes the averaged information $\avg({\bm Y}(k))$ to generate estimate $\hat{x}_k$, i.e., the conditional probability measures satisfy
	\begin{align}
	\mathbb{P}_f(\hat{x}_k\in\mathcal{A}|{\bm Y}(k)) = \mathbb{P}_f(\hat{x}_k\in\mathcal{A}|{\bm Y}'(k))
	\end{align}
	for any Borel set $\mathcal{A}\subseteq\R^n$ whenever $\avg({\bm Y}(k)) = \avg({\bm Y}'(k))$.
\end{definition}
Let $\mathcal F$ ($\mathcal F_{\rm c}$, resp.) be the set of all possible (compressed, resp.) estimators. In the following lemma, we show that it suffices to consider an estimator in $\mathcal F_{\rm c}$.
\begin{lemma} \label{lemma:compressed_estimator}
For any estimator $f\in\mathcal{F}$, there exists another compressed estimator $f'\in\mathcal{F}_{\rm c}$ such that for all $\delta>0$,
	\[e(f',k,\delta)  \leq e(f,k,\delta), \quad k=1,2,\ldots\]
\end{lemma}
\begin{proof}
	See Appendix~\ref{appendix:compressed}.
\end{proof}
\begin{remark}  \label{remark:sufficient}
	Intuitively, only measurements from benign sensors provide ``useful information'' needed to estimate the underlying state, while under the most harmful attack,  compromised sensors will merely generate disturbing noises.  In our case, the averaged information $\avg({\bf Y}(k))$ can \emph{fully} summarize the information contained in measurements from benign sensors due to the fact that $\avg({\bf Y}(k))$ is a sufficient statistic for the underlying state $x$ when there is no attacker. Therefore, it suffices to consider a compressed estimator that only utilizes the averaged information each time. This might be counterintuitive as one expects that with more information, i.e., using raw data ${\bf Y}(k)$, the compromised sensors could be detected more easily and, thus, better performance could be achieved. This, however, is not the case.
\end{remark}

\begin{remark} \label{remark:exponentialfamily}
Lemma~\ref{lemma:compressed_estimator} says that $\avg({\bf Y}(k))$ is a sufficient statistic for the underlying state $x$ whether or not the attacker is present. In fact, one may verify, using the same idea in Appendix~\ref{appendix:compressed}, in particular, the construction technique in~\eqref{eqn:constructf1}, that Lemma~\ref{lemma:compressed_estimator} holds if the distribution of $w_i(k)$ is in the exponential family and not necessarily Gaussian as we assume.  This is mainly due to the fact that, if the distribution of a one-shot observation is in the exponential family, the sufficient statistic of a set of i.i.d. observations is simply the sum of individual sufficient statistics, the size of which will not increase as data accumulate.  
\end{remark}

In the following, we refine the set $\mathcal{F}$ from another perspective. 
\begin{definition}
	An estimator $f$ is said to be deterministic w.r.t ${\bm Y}(k)$ if for every time $k$ and observations ${\bm Y}(k)$, the estimate $f({\bm Y}(k))$ is a single point in $\R^n$.
\end{definition}
Let $\mathcal F_{\rm d}$ be the set of all estimators that are deterministic w.r.t. ${\bm Y}(k)$. Then similar to the above lemma we have
\begin{lemma} \label{lemma:cd_estimator}
	For any estimator  $f\in\mathcal{F}$, there exists another deterministic one $f'\in\mathcal{F}_{\rm d}$ such that for all $\delta>0$
	\[r(f',\delta)  \geq r(f,\delta).\]
\end{lemma}
\begin{proof}
	See Appendix~\ref{appendix:cd_estimator}.
\end{proof}

Based on the above two lemmas, we further refine $\mathcal{F}$. 
\begin{definition} \label{def:deterministicandcompressed}
	An estimator $f$ is said to be compressed and deterministic if it is deterministic w.r.t. $\avg({\bm Y}(k))$, i.e., there exists a sequences of functions $\{\tilde{f}_k\}_{k=1,2,\ldots}$ with $\tilde{f}_k:\R^m\to\R^n$ such that the estimate at each time $k$
	\[f({\bm Y}(k)) = \tilde{f}_k(\avg({\bm Y}(k))).\]
\end{definition}


Let $\mathcal F_{\rm cd}$ be the set of all compressed and deterministic estimators. Obviously, $\mathcal{F}_{\rm cd}\subseteq \mathcal{F}_{\rm c}, \mathcal{F}_{\rm cd}\subseteq \mathcal{F}_{\rm d}$.
In the following theorem, we show that instead of $\mathcal F$, one may only consider the  set $\mathcal F_{\rm cd}$ for our problem. 
\begin{theorem} \label{theorem:restrict_estimator}
For any estimator $f\in\mathcal{F}$, there exists another compressed and deterministic estimator $f'\in\mathcal{F}_{\rm cd}$ such that
	\[r(f',\delta)  \geq r(f,\delta), \quad \forall \delta>0 .\]
\end{theorem}
\begin{proof}
	See Appendix~\ref{appendix:theoremRestEst}.
\end{proof}

\subsection{Optimal Estimator Based on Chebyshev Centers}
In this subsection, we propose an optimal compressed and deterministic estimator. To this end, we need the following definitions: The distance of a point $x_0\in\R^n$ to a bounded and non-empty set $\mathcal A\subseteq\R^n$ is defined as
\begin{align*}
\dist(x_0,\mathcal{A}) \triangleq \sup\{\|x-x_0\|_2: x\in\mathcal A\}.
\end{align*}
Moreover, the set's radius $\rad(\mathcal A)\in\R_+$ and Chebyshev center $\chv(\mathcal A)\in\R^n$ are defined by 
\begin{align}
\rad(\mathcal A) &\triangleq \min_{x_0\in\R^n} \dist(x_0,\mathcal{A}),    \label{eqn:radiusdef}\\
\chv(\mathcal A) &\triangleq \argmin_{x_0\in\R^n} \dist(x_0,\mathcal{A}).  \label{eqn:chebyshevdef}
\end{align}
Notice that the Chebyshev center exists and is unique, since $\R^n$ is uniformly convex and $\mathcal A$ is bounded~\cite[Part 5, \S 33]{holmes2006course}.

Given $y\in\R^m, x\in\R^n$, define their inconsistency  $d_x(y)$ as the optimal value of the following optimization problem:
\begin{mini}|l| 
	{a\in\R^m}{\frac{1}{2}\sum_{i=1}^m (y_i - H_ix +a_i)^2/W_i}{\label{eqn:opt1}}{}
	\addConstraint{\|a\|_0}{\leq q.}
\end{mini} 
Further define the set $\mathcal X(y,\phi), \phi\geq 0$ as the set of $x$ such that the inconsistency with $y$ is upper bounded by $\phi$, i.e.,
\begin{align} \label{eqn:bigX1}
	\mathcal X(y,\phi) \triangleq \{x\in\R^n: d_x(y)\leq \phi\}.
\end{align}
Given $\delta\geq 0$, define $\mathbb X(y,\delta)$ as the biggest $\mathcal X(y,\phi)$ of which the radius is upper bounded by $\delta$:
\begin{align} \label{eqn:Xestimator}
  \mathbb X(y,\delta) \triangleq \bigcup_{\rad( \mathcal X(y,\phi)) \leq \delta, \,\phi\geq 0} \mathcal X(y,\phi). 
\end{align}

It is easy to see that $\mathcal X(y,\phi)$ is monotonically increasing w.r.t. $\phi$. As a result, its radius is also increasing.  Notice also that given $y$, the radius $\rad( \mathcal X(y,\phi))$ is right-continuous with respect to $\phi$ (see details in Lemma~\ref{lemma:continuous} later). Therefore, it might happen that $\rad(\mathbb X(y,\delta)) < \delta$ for certain $\delta$, while in most cases $\rad(\mathbb X(y,\delta)) = \delta$ is achieved.    
Let $f^*_{\delta}$ be the estimator such that the estimate at time $k$ is the Chebyshev center of $\mathbb X(\avg({\bm Y}(k)),\delta)$, i.e., 
\begin{align} \label{eqn:optestimatork}
  f^*_{\delta}\left({\bm Y}(k)\right) = \chv\left(\mathbb X(\avg({\bm Y}(k)),\delta)\right).
\end{align}
For $y\in\R^m$ and $\delta>0$, we define $u(y, \delta)$ as the upper bound of the inconsistency between $y$ and the elements in $\mathbb{X}(y,\delta)$:
\begin{align} \label{eqn:upperd}
  u(y,\delta) \triangleq \sup_{x\in\mathbb{X}(y,\delta)} d_x(y).
\end{align}
With a slight abuse of notation, we define $u(\delta)$ as the lower bound of $u(y,\delta)$:
\begin{align} \label{eqn:upperd1}
  u(\delta) \triangleq \inf_{y\in\R^m} u(y,\delta).
\end{align}

We have our first main result about the estimator~\eqref{eqn:optestimatork}.
\begin{theorem}  \label{theorem:optimalestimator}
  Given any $\delta>0$, the estimator $f^*_{\delta}$ in~\eqref{eqn:optestimatork} is optimal in the sense that it maximizes the rate~\eqref{eqn:RateObj}, i.e., for any estimator $f\in\mathcal{F}$, 
  \begin{align} \label{eqn:theoremoptimalrate}
    r(f,\delta) \leq r(f^*_{\delta},\delta) = u(\delta).
  \end{align}
\end{theorem}
\begin{proof}
	See Appendix~\ref{proof:opt-estimator}.
\end{proof}
\begin{remark} \label{remark:parameterq}
	Notice that our estimator involves~$q$, as is the case in~\cite{mishra2017secure}, where the estimator (i.e., Algorithm~2 thereof)  depends on the perceived number of compromised sensors (or its upper bound) as well. On the contrary, estimators in~\cite{fawzi2014secure,mo2016secure} do not. In practice, the number of actually compromised sensors, $q_0$, might be smaller or larger than the design parameter $q$. If $q_0< q$, the performance of our estimator is lower bounded by $u(\delta)$ in~\eqref{eqn:upperd1}.  The details are as follows. With a little abuse of notation, in this remark, we use $d_x(y,q)$ (instead of $d_x(y)$) to denote the optimal value of optimization problem in~\eqref{eqn:opt1}, and rewrite $r(f^*_{\delta},\delta)$ as $r_q(f^*_{\delta},\delta)$.    Then the performance of our estimator when the number of compromised sensors is $q_0<q$ is:
	\begin{align}
		r_{q_0}(f^*_{\delta},\delta) =  \inf_{y\in\R^m}   \sup_{x\in\mathbb{X}(y,\delta)} d_x(y,q_0)  \geq   u(\delta).
	\end{align}
	We should admit that it is challenging to design an estimator that balances decently $r_q(f^*_{\delta},\delta)$ and $r_{q_0}(f^*_{\delta},\delta)$ in our case. Interested readers are referred to our previous work~\cite{xiaoqiang_securedetection}, where an detector that achieves the ``best" trade-off among performances with different $q$'s in the binary hypothesis testing case was provided.
	While if $q_0>q$, our estimator will be destroyed, i.e., $r(f^*_{\delta},\delta)=0$, as is the case in~\cite{mishra2017secure}. This is not desirable in practice. Our future work will investigate estimators independent of $q$.
\end{remark}

In the following lemma, we consider the case where Assumption~\ref{assumpt:2qobservable} is violated. 
\begin{lemma}  \label{lemma:vioassumpt2}
If Assumption~\ref{assumpt:2qobservable} is violated, the followings holds:
\begin{enumerate}
	\item For any $\delta>0$, there exists $y^*,x_1,x_2$ (dependent on $\delta$) such that $d_{x_1}(y^*)=d_{x_2}(y^*)=0$ and $\|x_1-x_2\|_2>\delta$;
	\item $ r(f,\delta) = r(f^*_{\delta},\delta) = 0$.
\end{enumerate}
\end{lemma}
\begin{proof}
	See Appendix~\ref{appendix:vioassumpt2}
\end{proof}
The above first bullet yields that for any $\delta>0$, there exists $y$ (dependent on $\delta$) such that $\mathbb X(y,\delta)$ is empty.  

\subsection{Numerical Implementation} \label{sec:numericalImple}
In this subsection, we provide an algorithm to compute the estimator $f^*_{\delta}$ proposed above. We shall first propose a method to compute the Chebyshev center and the radius of $\mathcal X(y,\phi)$ for a given $\phi$. This shares a similar spirit with~\cite{mo2015multi}. We then consider how to derive the appropriate $\phi$ using a modified bisection method. 
To proceed, we need the following definition and lemmas.

A variation of $d_x(y)$, where the support of $a$ in the definition in~\eqref{eqn:opt1} is given \textit{a priori}, is defined as follows:
\begin{definition}
  Given $x\in\R^n$, $y\in\R^m$, and index set $\mathcal{I}\subseteq\mathcal{M}$, the restricted inconsistency $d_x(y,\mathcal{I})$ is
  \begin{align} \label{eqn:distancerestict}
    d_x(y,\mathcal{I}) \triangleq \frac{1}{2}\sum_{i\in \mathcal{I}} (y_i - H_ix)^2/W_i.
  \end{align}	
\end{definition}
It is clear that with a fixed set $\mathcal{I}$, $d_x(y,\mathcal{I})$ is continuous w.r.t. both $x$ and $y$. Furthermore,
\begin{align*}
d_x(y) = \min_{\mathcal{I}\subseteq\mathcal{M},\,|\mathcal{I}| = m-q} d_x(y,\mathcal{I}).
\end{align*}

\begin{lemma} \label{lemma:restrictdistance}
  When $|\mathcal{I}|\geq m-2q$, the restricted inconsistency $d_x(y,\mathcal{I})$ can be equivalently written as:
  \begin{align} \label{eqn:distrewrite}
    d_x(y,{\mathcal I}) = (x-\kappa_\mathcal{I} y_{\mathcal{I}})^\top  \var(\mathcal{I}) (x- \kappa_\mathcal{I} y_{\mathcal{I}}) +\res(\mathcal{I})
  \end{align}
  where the ``variance''
  \begin{align}
    \var(\mathcal{I}) = \frac{1}{2}{\bm H}_\mathcal{I}^\top {\bm W}_{\{\mathcal{I}\}}^{-1}  {\bm H}_\mathcal{I} 
  \end{align}
  and the ``residue''
  \begin{align} \label{eqn:res}
    \res(\mathcal{I}) =  \frac{1}{2}(y_{\mathcal{I}} - {\bm H}_\mathcal{I}\kappa_\mathcal{I}y_{\mathcal{I}})^\top{\bm W}_{\{\mathcal{I}\}}^{-1}  (y_{\mathcal{I}} - {\bm  H}_\mathcal{I}\kappa_\mathcal{I}y_{\mathcal{I}})
  \end{align}
  with ${\bm W}_{\{\mathcal{I}\}}$ (different from ${\bm W}_{\mathcal{I}}$) being the square matrix obtained from   ${\bm W} = \diag(W_1,W_2,\ldots,W_m)$ after removing all of the rows and columns except those in the index set $\I$, and 
  \begin{align} \label{eqn:pesudoinverse}
    \kappa_\mathcal{I} = (\bm H_\mathcal{I}^\top {\bm W}_{\{\mathcal{I}\}}^{-1} \bm H_\mathcal{I} )^{-1}\bm H_\mathcal{I}^\top {\bm W}_{\{\mathcal{I}\}}^{-1} .
  \end{align}
\end{lemma}
\begin{proof}
	See Appendix~\ref{proof:lemmas}.
\end{proof}
In the following, we show that computing the Chebyshev center and radius of the set $\mathcal{X}(y,\phi)$ introduced in~\eqref{eqn:bigX1} can be transferred to a convex optimization problem. 
Notice that one can rewrite $\mathcal{X}(y,\phi)$ as:
\begin{align}
  \mathcal{X}(y,\phi) = \bigcup\mathcal{X}(y,\phi,\mathcal{I}),   \label{eqn:unionellipsoid}
\end{align}
where
\begin{align*}
 \mathcal{X}(y,\phi,\mathcal{I})\triangleq\{x\in\R^n: d_x(y,\mathcal{I}) \leq \phi\}.
\end{align*}
In other words, $\mathcal{X}(y,\phi)$ is a union of ellipsoids.
It is worth pointing out that if $\res(\mathcal{I}) = \phi$,  $\mathcal{X}(y,\phi,\mathcal{I})$ degenerates to a single point; and if $\res(\mathcal{I}) > \phi$,  $\mathcal{X}(y,\phi,\mathcal{I})$ is empty.  Therefore, to differentiate these cases, we define 
\begin{align*} 
\mathfrak{I}(\phi)& \triangleq\{\mathcal{I}\subseteq \mathcal{M}: \res(\mathcal{I})\leq \phi\text{ and } |\mathcal{I}| =  m-q\}, \addtag \label{eqn:collectionI} \\
  \mathfrak{I}_{+}(\phi) &\triangleq\{\mathcal{I}\subseteq \mathcal{M}: \res(\mathcal{I})< \phi\text{ and } |\mathcal{I}| =  m-q\}, \\
  \mathfrak{I}_{0}(\phi)&\triangleq  \mathfrak{I}(\phi)\setminus\mathfrak{I}_+(\phi). 
\end{align*}
\begin{lemma} \label{lemma:cvxopt}
  Given $\phi$ such that $\mathfrak{I}(\phi)$ is not empty.	Consider the following semidefinite programming problem:
  \begin{align*}
    \mathop{\minimize}_{\tau\in\R^{|\mathfrak{I}_{+}(\phi)|}, c, \psi\in\R} & \quad \psi \\
    \subjecto&  \quad \psi \geq 0,\\
                                                   &\quad \tau_i \geq 0, \:\forall 1\leq i \leq |\mathfrak{I}_{+}(\phi)|, \\
                                                   &\quad \tau_{\id(\mathcal{I})} \Theta(\mathcal{I},\phi) \succcurlyeq
                                                   \theta(c,\psi), \:\forall \mathcal{I}\in \mathfrak{I}_+(\phi), \\
                                                   &\quad   \begin{bmatrix}
                                                   \psi & (\kappa_\mathcal{I} y_{\mathcal{I}} - c)^\top \\
                                                   * & {\bm I}_n
                                                   \end{bmatrix}
                                                   \succcurlyeq 0, 	\:\forall \mathcal{I}\in \mathfrak{I}_{0}(\phi),	
  \end{align*}
  where 
  \begin{align*}
    &\Theta(\mathcal{I},\phi) \\
    \triangleq& 
                \begin{bmatrix}
                  \var(\mathcal{I})& -\var(\mathcal{I})\kappa_\mathcal{I} y_{\mathcal{I}} & 0 \\
                  * & (\kappa_\mathcal{I} y_{\mathcal{I}})^\top\var(\mathcal{I})\kappa_\mathcal{I} y_{\mathcal{I}} + \res(\mathcal{I}) -\phi & 0\\
                  0&0&0
                \end{bmatrix}
  \end{align*}
  with $*$ being recovered by symmetry, 
  \begin{align*}
  	\theta(c,\psi)  \triangleq \begin{bmatrix}
  	{\bm I}_n & -c & 0 \\
  	-c^\top & -\psi &c^\top \\
  	0 & c & -{\bm I}_n
  	\end{bmatrix},
  \end{align*} and $\id(\cdot): \mathfrak{I}_+(\phi) \mapsto \{1,2,\ldots,|\mathfrak{I}_{+}(\phi)|\}$ is any one-to-one function. Then
  \begin{align}
    \chv( \mathcal{X}(y,\phi) ) &= c^*,\\
    \rad( \mathcal{X}(y,\phi) ) &= \sqrt{\psi^*},
  \end{align}
  where $c^*$ and $\psi^*$ are the optimal solution of the semidefinite programming problem.
\end{lemma}
\begin{proof}
	See Appendix~\ref{proof:lemmas}.
\end{proof}

It follows from this lemma that, finding the Chebyshev center and radius of the set $\mathcal{X}(y,\phi)$ is  a semidefinite programming problem when $y,\phi$ are given. However, we are interested in finding the optimal estimator that maximize the rate $r(f,\delta)$, for a given $\delta$. In the following lemma, we give how $\rad(\mathcal{X}(y,\phi))$ varies with $\phi$, the illustration of which is in Section~\ref{sec:illustration}. 
\begin{lemma} \label{lemma:continuous}
  Given any $y\in\R^m$, the radius $\rad(\mathcal X(y,\phi))$ have the following properties:
  \begin{enumerate}
  \item $\rad(\mathcal X(y,\phi))$ is increasing, right-continuous w.r.t. $\phi$.
  \item If $\rad(\mathcal X(y,\phi))$ is discontinuous at a point $\phi_0$, then there must exist a set $\mathcal{I}\subseteq\mathcal{M}$ with $|\mathcal{I}| = m-q$ such that $\res(\mathcal{I}) = \phi_0.$
  \item When $\rad(\mathcal X(y,\phi))>0$, $\rad(\mathcal X(y,\phi))$ is strictly increasing w.r.t. $\phi$.
  \end{enumerate}
\end{lemma}
\begin{proof}
	See Appendix~\ref{proof:lemmas}.
\end{proof}

Given a predefined approximation bound $\varepsilon>0$, we compute the corresponding estimate $\hat{x}$ for an averaged measurement $\avg({\bm Y}(k))\in\R^m$ in Algorithm~\ref{alg:fstar}. Denoted by $\hat{f}_{\varepsilon}$ the resulting estimator, and by $\hat{f}(y,\varepsilon)$ the output of Algorithm~\ref{alg:fstar} (i.e., the estimate $\hat{x}$) when the inputs are $y,\varepsilon$.  

Notice that Algorithm~\ref{alg:fstar} is a slight variation of the classic bisection method. The distinguished part lies in~\eqref{eqn:guaranteenalg}, which together with Lemma~\ref{lemma:continuous} assures that 
 for any $y\in\R^m$, 
\[\inf_{x\not\in\mathcal{B}_\delta(\hat{f}(y,\varepsilon))} d_x(y) \geq u(y,\delta) - \varepsilon,\]
where $u(y,\delta)$ is defined in~\eqref{eqn:upperd}. Therefore, the following theorem readily follows:
\begin{theorem} \label{theorem:performanceguaAlg}
  Let an estimator $\hat{f}_{\varepsilon}({\bm Y}(k))=\hat{f}(\avg({\bm Y}(k)),\varepsilon)$ be computed by Algorithm~\ref{alg:fstar} with $\avg({\bm Y}(k)$ and $\varepsilon>0$ as inputs, then for all $\delta>0$ this estimator possesses the guaranteed performance: 
  \[r(\hat{f}_{\varepsilon},\delta) \geq r(f^*_{\delta},\delta) - \varepsilon, \]  
  where $f^*_{\delta}$ is the optimal estimator in~\eqref{eqn:optestimatork} 
\end{theorem}
Clearly, a smaller $\varepsilon$ in Algorithm~\ref{alg:fstar} leads to a better estimator, which, however, requires more iterations to run. 

\begin{algorithm}
  {\bf Inputs:} averaged measurements $y\in\R^m$,\\
  \hspace*{11mm} performance error tolerance $\varepsilon>0$.\\
  {\bf Output:} estimate $\hat{x}\in\R^n$\\
  {\bf Initialization:} Let
  \begin{align*}
    \underline{\phi}&=\min\{\res(\mathcal{I}): \mathcal{I}\subseteq\mathcal{M}, |\mathcal{I}|=m-q\} \\
    &\triangleq \Upsilon,
  \end{align*}
 and  $\overline{\phi}$ be such that $\rad(\mathcal{X}(y,\overline{\phi})) > \delta$. \\
  {\bf Repeat:}\\
  1. {\bf If}  $\overline{\phi} - \underline{\phi} < \varepsilon/2$ {\bf then} 
\begin{flalign*}
\hspace*{9mm}	&\phi = \max\{\Upsilon,\underline{\phi}-\varepsilon/2\},& \addtag \label{eqn:guaranteenalg}\\
\hspace*{9mm}	&\hat{x} = \chv(\mathcal{X}(y,\phi))&
\end{flalign*}
  \hspace*{7mm} {\bf Stop}\\
  \hspace*{2.5mm} {\bf EndIf} 

  2. Let $\phi = (\underline{\phi} + \overline{\phi})/2$.\\
  3. {\bf If}  $\rad(\mathcal{X}(y,\phi)) = \delta$ {\bf then} \\
  \hspace*{7mm} $\hat{x} = \chv(\mathcal{X}(y,\phi))$ \\
  \hspace*{7mm} {\bf Stop}\\
  \hspace*{2.5mm} {\bf ElseIf}  $\rad(\mathcal{X}(y,\phi)) > \delta$ {\bf then} \\
  \hspace*{7mm} $\overline{\phi} = \phi$ \\
  \hspace*{2.5mm} {\bf Else}  $\underline{\phi} = \phi$\\
  \hspace*{2.5mm} {\bf EndIf}

  \caption{Approximate Optimal Estimator $f^*_{\delta}$ in~\eqref{eqn:optestimatork}}
  \label{alg:fstar}
\end{algorithm}

\begin{remark}
Though semidefinite programming problem can be (approximately) solved in a polynomial time of program size~\cite{de2006aspects}. In our case, however, when $\phi$ is large enough, $|\mathfrak{I}_+(\phi)|=\begin{psmallmatrix}
	m\\q
\end{psmallmatrix}$, where $\begin{psmallmatrix}
m\\q
\end{psmallmatrix}$  is the binomial coefficient, which renders the optimization problem in Lemma~\ref{lemma:cvxopt} rather computationally heavy when $\phi$ and $m$ are large. Nevertheless, we defend our estimator from the following two aspects. First, though efficient and optimal algorithms might exist in certain problems, see e.g.,~\cite{xiaoqiang_securedetection}, the resilient information fusion under sparse attack is intrinsically of combinatorial nature, see e.g.,~\cite{fawzi2014secure, mishra2017secure, chong2015observability}, since we basically need to search over all combinations of possibly healthy sensors. Nevertheless, this work is just a starting point, and we are planning to investigate the approaches that could relieve the computational burden (in certain cases) just as in~\cite{shoukry2017secure, lee2018redundant, nakahira2018attack}. Second, in practice, a small $\delta$ would be usually chosen. Then the size of $\mathfrak{I}_+(\phi)$ will be small as well no matter how big $m$ is, and, therefore, the optimization problem in Lemma~\ref{lemma:cvxopt} could be efficiently solved. Though finding $\Upsilon$ of Algorithm~\ref{alg:fstar} is of combinatorial nature, computing $\res(\mathcal{I})$ (given in~\eqref{eqn:res}) for a given set $\mathcal{I}$ is light (notice that ${\bm W}_{\{\mathcal{I}\}}$ is a diagonal matrix and its inverse, therefore, is readily given). Therefore, the computation burden of Algorithm~\ref{alg:fstar} could be tolerated for a large $m$.
\end{remark}
\begin{remark} \label{remark:resilience}
	 The resilience of the proposed optimal estimator $f^*_{\delta}$ in~\eqref{eqn:optestimatork} may not be that apparent since Chebyshev center itself is sensitive to noises, i.e., the Chebyshev center of a set $\mathcal{A}$ can be driven to anywhere even if only one point of $\mathcal{A}$ is allowed to be manipulated.
	Nevertheless, the resilience of the estimator $f^*_{\delta}$ can be heuristically explained by the following two factors. First, when the time $k$ is large enough, the measurements from benign sensors can lead to rather accurate estimate, i.e., the $\res(\mathcal{I}_*)$  will be quite small, where $\mathcal{I}_*$ is of size $m-q$ and contains no compromised sensors. Therefore, $\mathcal{I}_*$ would be in the collection $\mathfrak{I}(\phi)$ defined in~\eqref{eqn:collectionI} and somehow serves as an anchor when computing the estimate as in Algorithm~\ref{alg:fstar} and Lemma~\ref{lemma:cvxopt}. Second, when the injected bias of a compromised sensor is too large, the resulting $\res(\mathcal{I})$ for any $\mathcal{I}$ containing this compromised sensor will be quite large as well. Therefore, the set $\mathcal{I}$ will not be in the collection  $\mathfrak{I}(\phi)$ and the measurement from this compromised sensor will be discarded when computing the estimate.
\end{remark}
\section{Uniformly Optimal Estimator for Homogeneous Sensors} \label{sec:uniformoptimal}
In this section we provide a simple yet uniformly optimal estimator $f$ such that $r(f,\delta)$ is simultaneously maximized for all $\delta >0$ when the sensors are homogeneous, i.e., $H_1=\cdots =H_m$ and $W_1=\cdots=W_m$. Notice that when homogeneous sensors are considered, to satisfy the $2q$-observable assumption in Assumption~\ref{assumpt:2qobservable}, the state has to be scalar, i.e., $x\in \R$.

To proceed, we first provide an upper bound of the optimal performance proved in Theorem~\ref{theorem:optimalestimator}, $u(\delta)$,  for any $\delta$ and \emph{any} system models (instead of only homogeneous sensors).

\begin{lemma} \label{theorem:upperbound}
	The optimal performance $u(\delta)$ in~\eqref{eqn:theoremoptimalrate} is upper bounded:
	\begin{align*}
		u(\delta) \leq \bar{u}(\delta),
	\end{align*}
	where $\bar{u}(\delta) = \delta^2\bar{u}(1)$ with $\bar{u}(1)$ being the optimal value of the following optimization problem:
	\begin{mini}|l| 
		{x\in\R^n, s\in\R^m}{\frac{1}{2}\sum_{i=1}^m (H_ix +s_i)^2/W_i}{\label{eqn:optupper}}{}
		\addConstraint{\|s\|_0}{ \leq 2q}
		\addConstraint{\|x\|_2}{= 1.}
	\end{mini} 
\end{lemma}
\begin{proof}
	See Appendix~\ref{appendix:upperbound}.
\end{proof}

In the remainder of this section, we consider the case where sensors are homogeneous and the system is scalar. Then without loss of generality, we let $H_i = 1$ for any $1\leq i\leq m$. We define the estimator $f^{\trm}$ be the ``trimmed mean", i.e.,
\begin{align} \label{eqn:ftrim}
f^{\trm}\left({\bm Y}(k)\right) = \trm\left(\avg({\bm Y}(k))\right),
\end{align}
where for any $y\in\R^m$, 
\begin{align}
	\trm(y) \triangleq \frac{1}{m-2q} \sum_{i=q+1}^{m-q} y_{[i]}
\end{align}
with $y_{[i]}$ being the $i$-th smallest element. In other words, $\trm(y)$ first discards the largest $q$ and smallest $q$ elements of $y$, and then averages over the remaining ones.

We show that the trimmed mean estimator $f^{\trm}$ is uniformly optimal in  Theorem~\ref{theorem:uniform}. The theorem is proved by showing that $f^{\trm}$ achieves the upper bound in Lemma~\ref{theorem:upperbound} for every $\delta$, which, in turn, means that the upper bound is tight when homogeneous sensors are considered.      
\begin{theorem} \label{theorem:uniform}
	When the sensors are homogeneous (and thus the system state is scalar), $f^{\trm}$ in~\eqref{eqn:ftrim} is uniformly optimal, i.e.,
	\[r(f^{\trm},\delta) = u(\delta)\]
	holds for every $\delta$. 
\end{theorem}
\begin{proof}
	See Appendix~\ref{appendix:uniform}.
\end{proof}
%
%

\section{Numerical Examples}
\label{sec:simulation}

\subsection{Illustration of $\mathcal X(y,\phi)$ and $\rad(\mathcal X(y,\phi))$} \label{sec:illustration}
We illustrate how  $\mathcal X(y,\phi)$ and  $\rad(\mathcal X(y,\phi))$ vary with $\phi$ in Fig.~\ref{Fig:ellipse} and Fig.~\ref{Fig:radius}, respectively. The parameters used are summarized as follows: $m=4$ sensors used to estimate $x\in\R^2$, $q=1$ sensor might be manipulated, measurement matrix ${\bm H}=[1, 0;0, 1;1, 2;2, 1]$, covariance matrix ${\bm W}= \diag(1,2,2,1)$, and observation $y=[4;-4;5;-5]$. Let $\res_{[i]}$ be the $i$-th item of the set $\{\res(\mathcal{I}): \mathcal{I}\subseteq \mathcal{M}, |\mathcal{I}| = m-q\}$ sorted in an ascending order. Then we have, in our case, that $\res_{[1]}=3.68182,  \res_{[2]}= 5.78571, \res_{[3]}=13.5, \res_{[4]}=24.3$.

From Fig.~\ref{Fig:ellipse}, one sees that $\mathcal X(y,\phi)$ is indeed a union of several ellipses. One also sees in Fig.~\ref{Fig:radius} that $\rad(\mathcal X(y,\phi))$ is strictly increasing w.r.t. $\phi$ when $\rad(\mathcal X(y,\phi))>0$, and discontinuous only at $\res_{[2]}$ and $\res_{[3]}$, which verifies Lemma~\ref{lemma:continuous}. Notice also that as $\phi$ crosses $\res_{[4]}$ from below, the new ellipse, which is indicated by the red one in the right-below sub-figure of Fig.~\ref{Fig:ellipse}, is inside the blue dashed circle that covers the previous three ellipses. Therefore, $\rad(\mathcal X(y,\phi))$ is continuous at $\phi=\res_{[4]}$.

	\begin{figure}[ht]
		\begin{center}
		\input{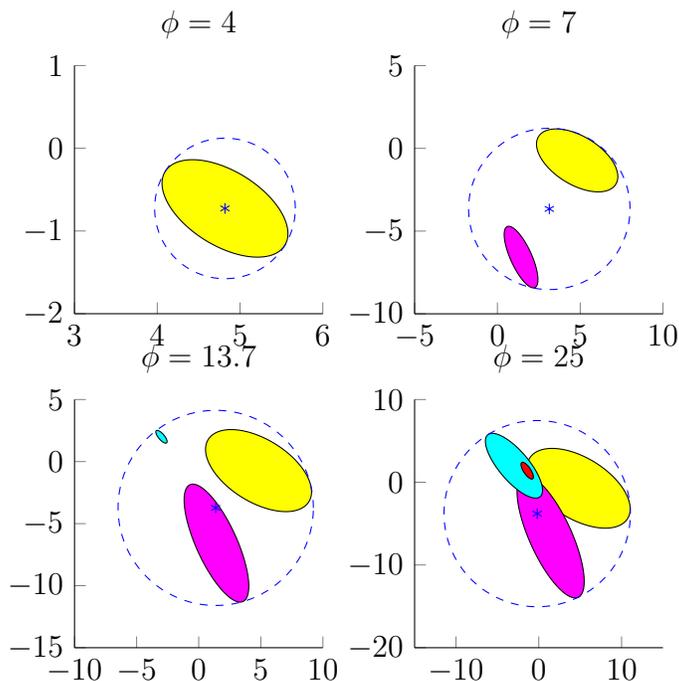}
		\caption{The set $\mathcal X(y,\phi)$ with different $\phi$'s. In each of the four sub-figures, x-axis is $x_1$ and y-axis $x_2$. The filled area corresponds to $\mathcal X(y,\phi)$. The blue ``*" is the Chebyshev center of $\mathcal X(y,\phi)$, and blue dashed line the circle centered at the Chebyshev center with radius $\rad(\mathcal X(y,\phi))$. }
		\label{Fig:ellipse}
	\end{center}
	\end{figure}

	\begin{figure}[ht]
		\begin{center}
%
%
\begin{tikzpicture}

\begin{axis}[%
width=2.8in,
height=1.8in,
at={(0.663in,0.442in)},
scale only axis,
xmin=0,
xmax=30,
xlabel style={font=\color{white!15!black}},
xlabel={$\phi$},
ymin=0,
ymax=14,
ylabel style={font=\color{white!15!black}},
ylabel={$\rad(\mathcal X(y,\phi))$},
axis background/.style={fill=white}
]
\addplot [color=blue, forget plot]
  table[row sep=crcr]{%
3.68181818181818	0.000481885\\
3.78181818181818	0.476\\
3.88181818181818	0.673358\\
3.98181818181818	0.824929\\
4.08181818181818	0.952326\\
4.18181818181818	1.06462\\
4.28181818181818	1.16625\\
4.38181818181818	1.26002\\
4.48181818181818	1.34701\\
4.58181818181818	1.42871\\
4.68181818181818	1.50598\\
4.78181818181818	1.57946\\
4.88181818181818	1.64941\\
4.98181818181818	1.71709\\
5.08181818181818	1.78158\\
5.18181818181818	1.84415\\
5.28181818181818	1.90466\\
5.38181818181818	1.96362\\
5.48181818181818	2.02051\\
5.58181818181818	2.07585\\
5.68181818181818	2.12979\\
5.78181818181818	2.18201\\
};
\addplot [color=blue, forget plot]
  table[row sep=crcr]{%
5.88181818181818	4.1691\\
5.98181818181818	4.26904\\
6.08181818181818	4.35147\\
6.18181818181818	4.42414\\
6.28181818181818	4.49043\\
6.38181818181818	4.55212\\
6.48181818181818	4.61029\\
6.58181818181818	4.66564\\
6.68181818181818	4.71864\\
6.78181818181818	4.76966\\
6.88181818181818	4.81896\\
6.98181818181818	4.86675\\
7.08181818181818	4.91321\\
7.18181818181818	4.95847\\
7.28181818181818	5.00265\\
7.38181818181818	5.04584\\
7.48181818181818	5.08812\\
7.58181818181818	5.12956\\
7.68181818181818	5.17022\\
7.78181818181818	5.21016\\
7.88181818181818	5.24943\\
7.98181818181818	5.28805\\
8.08181818181818	5.32608\\
8.18181818181818	5.36355\\
8.28181818181818	5.40048\\
8.38181818181818	5.4369\\
8.48181818181818	5.47284\\
8.58181818181818	5.50833\\
8.68181818181818	5.54337\\
8.78181818181818	5.57799\\
8.88181818181818	5.61222\\
8.98181818181818	5.64605\\
9.08181818181818	5.67952\\
9.18181818181818	5.71263\\
9.28181818181818	5.7454\\
9.38181818181818	5.77784\\
9.48181818181818	5.80996\\
9.58181818181818	5.84177\\
9.68181818181818	5.87329\\
9.78181818181818	5.90451\\
9.88181818181818	5.93546\\
9.98181818181818	5.96614\\
10.0818181818182	5.99655\\
10.1818181818182	6.02672\\
10.2818181818182	6.05663\\
10.3818181818182	6.0863\\
10.4818181818182	6.11574\\
10.5818181818182	6.14495\\
10.6818181818182	6.17394\\
10.7818181818182	6.20271\\
10.8818181818182	6.23127\\
10.9818181818182	6.25963\\
11.0818181818182	6.28779\\
11.1818181818182	6.31575\\
11.2818181818182	6.34352\\
11.3818181818182	6.3711\\
11.4818181818182	6.3985\\
11.5818181818182	6.42572\\
11.6818181818182	6.45277\\
11.7818181818182	6.47964\\
11.8818181818182	6.50635\\
11.9818181818182	6.53289\\
12.0818181818182	6.55928\\
12.1818181818182	6.5855\\
12.2818181818182	6.61157\\
12.3818181818182	6.63749\\
12.4818181818182	6.66327\\
12.5818181818182	6.68889\\
12.6818181818182	6.71437\\
12.7818181818182	6.73972\\
12.8818181818182	6.76492\\
12.9818181818182	6.78999\\
13.0818181818182	6.81492\\
13.1818181818182	6.83973\\
13.2818181818182	6.8644\\
13.3818181818182	6.88895\\
13.4818181818182	6.91338\\
};
\addplot [color=blue, forget plot]
  table[row sep=crcr]{%
13.5818181818182	7.76251\\
13.6818181818182	7.85003\\
13.7818181818182	7.92744\\
13.8818181818182	7.99722\\
13.9818181818182	8.06122\\
14.0818181818182	8.12076\\
14.1818181818182	8.1768\\
14.2818181818182	8.23001\\
14.3818181818182	8.2809\\
14.4818181818182	8.32982\\
14.5818181818182	8.37707\\
14.6818181818182	8.42285\\
14.7818181818182	8.46734\\
14.8818181818182	8.51068\\
14.9818181818182	8.55299\\
15.0818181818182	8.59435\\
15.1818181818182	8.63485\\
15.2818181818182	8.67456\\
15.3818181818182	8.71355\\
15.4818181818182	8.75185\\
15.5818181818182	8.78953\\
15.6818181818182	8.82661\\
15.7818181818182	8.86314\\
15.8818181818182	8.89914\\
15.9818181818182	8.93465\\
16.0818181818182	8.96969\\
16.1818181818182	9.00429\\
16.2818181818182	9.03847\\
16.3818181818182	9.07225\\
16.4818181818182	9.10564\\
16.5818181818182	9.13866\\
16.6818181818182	9.17133\\
16.7818181818182	9.20366\\
16.8818181818182	9.23567\\
16.9818181818182	9.26736\\
17.0818181818182	9.29876\\
17.1818181818182	9.32986\\
17.2818181818182	9.36068\\
17.3818181818182	9.39123\\
17.4818181818182	9.42152\\
17.5818181818182	9.45156\\
17.6818181818182	9.48135\\
17.7818181818182	9.5109\\
17.8818181818182	9.54021\\
17.9818181818182	9.56931\\
18.0818181818182	9.59818\\
18.1818181818182	9.62684\\
18.2818181818182	9.6553\\
18.3818181818182	9.68355\\
18.4818181818182	9.71161\\
18.5818181818182	9.73947\\
18.6818181818182	9.76715\\
18.7818181818182	9.79464\\
18.8818181818182	9.82195\\
18.9818181818182	9.8491\\
19.0818181818182	9.87607\\
19.1818181818182	9.90287\\
19.2818181818182	9.92951\\
19.3818181818182	9.95599\\
19.4818181818182	9.98232\\
19.5818181818182	10.0085\\
19.6818181818182	10.0345\\
19.7818181818182	10.0604\\
19.8818181818182	10.0861\\
19.9818181818182	10.1117\\
20.0818181818182	10.1372\\
20.1818181818182	10.1625\\
20.2818181818182	10.1877\\
20.3818181818182	10.2127\\
20.4818181818182	10.2376\\
20.5818181818182	10.2624\\
20.6818181818182	10.2871\\
20.7818181818182	10.3116\\
20.8818181818182	10.3361\\
20.9818181818182	10.3604\\
21.0818181818182	10.3846\\
21.1818181818182	10.4086\\
21.2818181818182	10.4326\\
21.3818181818182	10.4565\\
21.4818181818182	10.4802\\
21.5818181818182	10.5039\\
21.6818181818182	10.5274\\
21.7818181818182	10.5508\\
21.8818181818182	10.5741\\
21.9818181818182	10.5974\\
22.0818181818182	10.6205\\
22.1818181818182	10.6435\\
22.2818181818182	10.6664\\
22.3818181818182	10.6893\\
22.4818181818182	10.712\\
22.5818181818182	10.7347\\
22.6818181818182	10.7572\\
22.7818181818182	10.7797\\
22.8818181818182	10.8021\\
22.9818181818182	10.8243\\
23.0818181818182	10.8465\\
23.1818181818182	10.8687\\
23.2818181818182	10.8907\\
23.3818181818182	10.9126\\
23.4818181818182	10.9345\\
23.5818181818182	10.9563\\
23.6818181818182	10.978\\
23.7818181818182	10.9996\\
23.8818181818182	11.0211\\
23.9818181818182	11.0426\\
24.0818181818182	11.064\\
24.1818181818182	11.0853\\
24.2818181818182	11.1065\\
24.3818181818182	11.1277\\
24.4818181818182	11.1488\\
24.5818181818182	11.1698\\
24.6818181818182	11.1907\\
24.7818181818182	11.2116\\
24.8818181818182	11.2324\\
24.9818181818182	11.2531\\
25.0818181818182	11.2738\\
25.1818181818182	11.2944\\
25.2818181818182	11.3149\\
25.3818181818182	11.3354\\
25.4818181818182	11.3558\\
25.5818181818182	11.3761\\
25.6818181818182	11.3964\\
25.7818181818182	11.4166\\
25.8818181818182	11.4367\\
25.9818181818182	11.4568\\
26.0818181818182	11.4768\\
26.1818181818182	11.4967\\
26.2818181818182	11.5166\\
26.3818181818182	11.5365\\
26.4818181818182	11.5562\\
26.5818181818182	11.5759\\
26.6818181818182	11.5956\\
26.7818181818182	11.6152\\
26.8818181818182	11.6347\\
26.9818181818182	11.6542\\
27.0818181818182	11.6736\\
27.1818181818182	11.693\\
27.2818181818182	11.7123\\
27.3818181818182	11.7316\\
27.4818181818182	11.7508\\
27.5818181818182	11.7699\\
27.6818181818182	11.789\\
27.7818181818182	11.808\\
27.8818181818182	11.827\\
27.9818181818182	11.846\\
28.0818181818182	11.8649\\
28.1818181818182	11.8837\\
28.2818181818182	11.9025\\
28.3818181818182	11.9212\\
28.4818181818182	11.9399\\
28.5818181818182	11.9585\\
28.6818181818182	11.9771\\
28.7818181818182	11.9957\\
28.8818181818182	12.0141\\
28.9818181818182	12.0326\\
29.0818181818182	12.051\\
29.1818181818182	12.0693\\
29.2818181818182	12.0876\\
};
\addplot [color=black, dashed, forget plot]
  table[row sep=crcr]{%
5.78571	0\\
5.78571	14\\
};
\addplot [color=black, dashed, forget plot]
  table[row sep=crcr]{%
3.68182	0\\
3.68182	14\\
};
\addplot [color=black, dashed, forget plot]
  table[row sep=crcr]{%
13.5	0\\
13.5	14\\
};
\addplot [color=black, dashed, forget plot]
  table[row sep=crcr]{%
24.3	0\\
24.3	14\\
};

\addplot [color=blue,  forget plot]
table[row sep=crcr]{%
	0	0\\
	3.68182	0\\
};

\end{axis}

\node [above right] at (7,	5.65) {$ \res_{[4]}$};
\node [above right] at (4.55,	5.65) {$ \res_{[3]}$};
\node [above right] at (2.1,	5.65) {$ \res_{[1]}$};
\node [above right] at (2.9,	5.65) {$ \res_{[2]}$};
\end{tikzpicture}%
		\caption{Radius $\rad(\mathcal X(y,\phi))$ as a function of $\phi$. }
		\label{Fig:radius}
	\end{center}
	\end{figure}
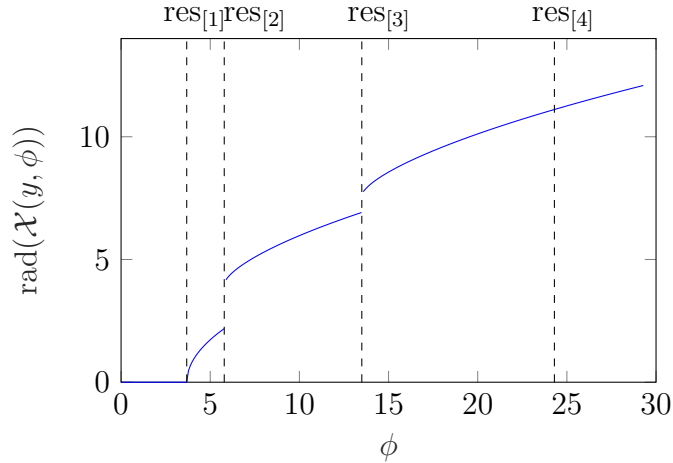

\subsection{Resilience of the Proposed Estimator}
In the following, in order to verify the intuitive comments of Remark~\ref{remark:resilience} about the resilience of $f^*_{\delta}$~\eqref{eqn:optestimatork}, we use a numerical example to show how the output of $f^*_{\delta}$ varies with the injected bias and $\delta$. The parameters used are summarized as follows: $m=4$ sensors used to estimate $x\in\R^2$, measurement matrix ${\bm H}=[1, 0;0, 1;1, 2;2, 1]$, covariance matrix ${\bm W}= \diag(1,2,2,1)$, and observation $z=[1;1;3;3]$. 
We let the fourth sensor be attacked, i.e., the first three elements of $y$ are $[1;1;3]$ and $y_4 = z_4+ a$. In particular, we let $a$ vary from $0$ to $15$. We simulate our estimator $f^*_\delta$ for two different error thresholds $\delta=1,3$, and further compare it to the least squares estimator, which computes the estimate as $({\bm H}^\top {\bm W}^{-1}{\bm H})^{-1}{\bm H}^\top{\bm W}^{-1}y$. When using Algorithm~\ref{alg:fstar}, we let the performance error tolerance $\varepsilon=0.001$. 

The result is illustrated in Fig.~\ref{Fig:resilience}. One sees that when the bias injected $a$ is too large, the estimation error of our algorithm is zero, i.e., the attack effects are eliminated. This is consistent with intuitive comments in Remark~\ref{remark:resilience}. Furthermore, when using a smaller $\delta$ (i.e., $\delta=1$ in our example), the estimator tends to discard the injected bias: the zero-error range is $a\in[3,\infty]$ for $\delta=1$, which is contrasted with $[8,\infty]$ for $\delta=3$. This is because given the same observation $y$, smaller $\delta$ is, smaller $\phi$ and, thus, the collection $\mathfrak{I}(\phi)$ are, which means that the ``abnormal" data (with large $\res(\mathcal{I})$) will be more likely to be discarded. It is clear that the naive least squares estimator is not resilient to the attack.

	\begin{figure}[ht]
		\begin{center}
%
%
\definecolor{mycolor1}{rgb}{0.00000,0.44700,0.74100}%
\definecolor{mycolor2}{rgb}{0.85000,0.32500,0.09800}%
\definecolor{mycolor3}{rgb}{0.92900,0.69400,0.12500}%
\begin{tikzpicture}

\begin{axis}[%
width=3.2in,
height=2.21in,
at={(0.758in,0.481in)},
scale only axis,
xmin=0,
xmax=15,
ymin=0,
ymax=6,
axis background/.style={fill=white},
legend style={at={(0.02,0.7)}, anchor=south west, legend cell align=left, align=left, draw=white!15!black}
]
\addplot [color=mycolor1, mark=asterisk, mark options={solid, mycolor1}]
  table[row sep=crcr]{%
0 5.44216716124729e-06\\
1 0.404051211926202\\
2 0.924989169616453\\
3 0.946362957190796\\
4 1.07965025203601\\
5 1.32993544913381\\
6 1.52350715884455\\
7 1.57630060110592\\
8 2.32581933568804e-06\\
9 6.38003949470632e-06\\
10 7.01420223370537e-06\\
11 1.14248268329468e-06\\
12 1.67981367673645e-06\\
13 0\\
14 0\\
15 0\\
};
\addlegendentry{$f^*_\delta, \delta=3$}

\addplot [color=mycolor2, mark=o, mark options={solid, mycolor2}]
  table[row sep=crcr]{%
0 2.81091600305088e-05\\
1 0.315182188116361\\
2 0.508161122970431\\
3 1.33388437086458e-05\\
4 1.60911469869570e-05\\
5 1.74556094245161e-05\\
6 1.74550194795729e-05\\
7 1.74531860045418e-05\\
8 1.74552832422059e-05\\
9 1.74515115845250e-05\\
10 1.74544096882810e-05\\
11 1.74560570079828e-05\\
12 1.74552832422059e-05\\
13 0\\
14 0\\
15 0\\
};
\addlegendentry{$f^*_\delta, \delta=1$}

\addplot [color=mycolor3, mark=square, mark options={solid, mycolor3}]
  table[row sep=crcr]{%
0	2.22044604925031e-16\\
1	0.393280865770661\\
2	0.786561731541322\\
3	1.17984259731198\\
4	1.57312346308264\\
5	1.96640432885331\\
6	2.35968519462397\\
7	2.75296606039463\\
8	3.14624692616529\\
9	3.53952779193595\\
10	3.93280865770661\\
11	4.32608952347727\\
12	4.71937038924793\\
13	5.11265125501859\\
14	5.50593212078925\\
15	5.89921298655991\\
};
\addlegendentry{Least squares estimator}

\end{axis}
\end{tikzpicture}%
		\caption{$2$-norm of estimation error as a function of the bias injected.  Our algorithm with different $\delta$ and the least squares estimator are compared. }
		\label{Fig:resilience}
	\end{center}
	\end{figure}
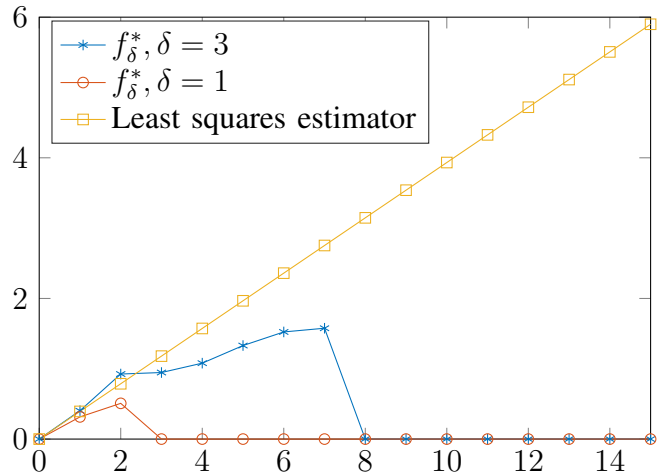

\subsection{Comparison with Other Estimators}
In this section, we compare our estimator with the LASSO.
In our case, given $\avg({\bm Y}(k))=y\in\R^m$, the LASSO reads 
\begin{align} \label{eqn:lasso}
\mathop{\rm minimize}_{x\in\R^n, a\in\R^m}  \|({\bm W}/k)^{-1/2}(y-{\bm H}x-a)\|^2 + \lambda\|a\|_1,
\end{align}
where $\lambda$ is predefined parameter and the optimal solution $x$ is the estimate.  Basically, the smaller $\lambda$ is, the securer is LASSO. Therefore, in our simulation, we set $\lambda=10^{-3}$.
Notice that the $l_0$ and $l_1$-based state estimation procedures~\cite{fawzi2014secure,pajic2017attack} works in systems without noises or with (small) bounded measurement noises, while the estimator in~\cite{mishra2017secure} (i.e., Algorithm~2 thereof) is undecided for (many) certain observations, that is, it can happen that no subset of sensors are deemed as attack free and, therefore, no output will be generated. We should also note that while~\cite{mo2016secure} proves the resilience of LASSO when each sensor is observable, i.e., $H_i$ is scalar in our case, the LASSO under sparse  attack is not resilient in general; see~\cite{HanMX15}.  Therefore, we consider scalar state in this simulation and for simplicity, we further assume the sensors are homogeneous.

We assume there are totally $m=5$ sensors, of which $q=1$ sensor is compromised. We let measurement matrix ${\bm H}=[1; 1;1;1;1]$ and covariance matrix ${\bm W}= \diag(1,1,1,1,1)$. When computing the worst-case probability $e(f,k,\delta)$ in~\eqref{eqn:errorProb}, we assume that, without loss of generality, the true state is $x=0$ and the fifth sensor compromised. We then simulate the error probability for a fixed $y_5$ with   
 $y=\avg({\bm Y}(k))$ being the averaged measurement, the maximum of which is then regarded as the worst-case probability $e(f,k,\delta)$.
From Fig.~\ref{Fig:comp}, one sees that for either $\delta=1$ or $\delta=1.5$, the performances of $f^{\trm}$ and $f^*_{\delta}$ are quite close, which is consistent with the uniform optimality of $f^{\trm}$ stated in Theorem~\ref{theorem:uniform}. One should also note that both $f^{\trm}$ and $f^*_{\delta}$ outperform the LASSO. 
 

	\begin{figure}[ht]
		\begin{center}
%
%
\definecolor{mycolor1}{rgb}{0.00000,0.44700,0.74100}%
\definecolor{mycolor2}{rgb}{0.85000,0.32500,0.09800}%
\definecolor{mycolor3}{rgb}{0.92900,0.69400,0.12500}%
\definecolor{mycolor4}{rgb}{0.49400,0.18400,0.55600}%
\definecolor{mycolor5}{rgb}{0.46600,0.67400,0.18800}%
\definecolor{mycolor6}{rgb}{0.30100,0.74500,0.93300}%
\begin{tikzpicture}

\begin{axis}[%
width=2.8in,
height=1.8in,
at={(0.663in,0.442in)},
scale only axis,
xmin=1,
xmax=10,
xtick={ 1,  5, 10},
xlabel style={font=\color{white!15!black}},
xlabel={time},
ymode=log,
ymin=1e-20,
ymax=1,
yminorticks=true,
ylabel style={font=\color{white!15!black}},
ylabel={worst-case probability},
axis background/.style={fill=white},
legend style={at={(0.02,-0.1)}, anchor=south west, legend cell align=left, align=left, draw=white!15!black}
]
\addplot [color=mycolor1, mark=asterisk, mark options={solid, mycolor1}]
  table[row sep=crcr]{%
1	0.122245\\
5	0.00018861\\
10	6.10543e-08\\
};
\addlegendentry{$f^*_\delta,\delta=1$}

\addplot [color=mycolor2, mark=o, mark options={solid, mycolor2}]
  table[row sep=crcr]{%
1	0.0961741\\
5	0.0001431\\
10	4.3281e-08\\
};
\addlegendentry{$f^{\trm},\delta=1$}

\addplot [color=mycolor3, mark=square, mark options={solid, mycolor3}]
  table[row sep=crcr]{%
1	0.139887\\
5	0.00063891\\
10	1.4835e-06\\
};
\addlegendentry{LASSO, $\delta=1$}

\addplot [color=mycolor4, mark=asterisk, mark options={solid, mycolor4}]
  table[row sep=crcr]{%
1	0.0178203\\
5	4.76865e-09\\
10	1.35969e-16\\
};
\addlegendentry{$f^*_\delta,\delta=1.5$}

\addplot [color=mycolor5, mark=o, mark options={solid, mycolor5}]
  table[row sep=crcr]{%
1	0.0134201\\
5	3.58253e-09\\
10	1.02141e-16\\
};
\addlegendentry{$f^{\trm},\delta=1.5$}

\addplot [color=mycolor6, mark=square, mark options={solid, mycolor6}]
  table[row sep=crcr]{%
1	0.0198031\\
5	1.60569e-07\\
10	1.10589e-12\\
};
\addlegendentry{LASSO, $\delta=1.5$}

\end{axis}
\end{tikzpicture}%
		\caption{Worst-case probability $e(f,k,\delta)$ in~\eqref{eqn:errorProb} as a function of estimator $f$ (our proposed estimator $f^*_{\delta}$ in Algorithm~\ref{alg:fstar}, trimmed mean estimator $f^{\trm}$ in~\eqref{eqn:ftrim}, and LASSO in~\eqref{eqn:lasso}), time $k$ ($1,5,10$), and error threshold $\delta$ ($1,1.5$). }
		\label{Fig:comp}
	\end{center}
	\end{figure}
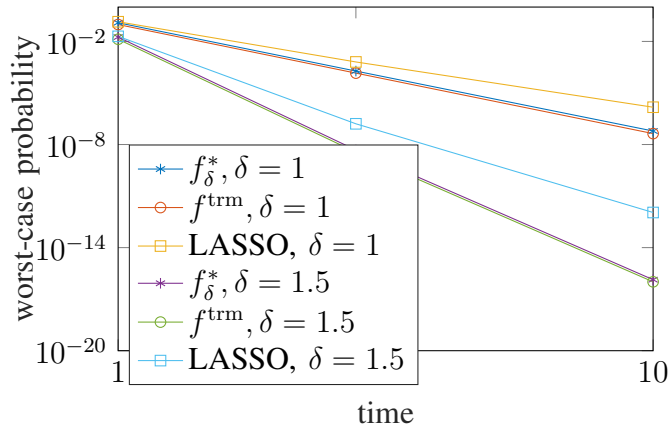

\section{Conclusion and Future Work}
\label{sec:conclusion}
In this paper, we provided a different perspective on secure static state estimation with Byzantine sensors by introducing  a new probabilistic performance metric, i.e., the decaying rate of the worst-case probability that the estimation error is larger than some value $\delta$ rather than the worst-case error or estimation error covariance in the existing literature.  This new metric does not necessarily require bounded noise. With this metric, we gave an optimal estimator for any given error threshold $\delta$, which is the Chebyshev center of a certain set, and proposed an algorithm to compute it. A significant byproduct is that if  distribution of the observation noise is in the exponential family, the sufficient statistic for the underlying state remains the same whether or not the attacker is present. When the sensors are homogeneous, we further derived a  simple yet uniformly optimal estimator, which, to be specific, is the trimmed mean of the averaged observations and simultaneously optimal for every $\delta$. 

For the future work, there are two interesting directions. One is to extend the existing results into dynamic systems, while the other one is to investigate the uniformly optimal estimator when sensors are heterogeneous.

\appendices

\section{Proof of Lemma~\ref{lemma:compressed_estimator}}
\label{appendix:compressed}
The proof is of constructive nature and mainly stems from the fact that  $\avg({\bm Z}(k))$ is a sufficient statistic for the underlying state $x$, where ${\bm Z}(k)$ is the ``true`` measurement matrix when there are no attacks and is defined in the same manner with ${\bm Y}(k)$. 

In the following, for simplicity of presentation, we do not distinguish the probability density function (pdf) for a continuous random variable and probability mass function (pmf) for a discrete one. Therefore, in some cases the summation is actually needed though we use integration universally.

For any $f\in\mathcal{F}$,  we let $f'$ satisfy~\eqref{eqn:constructf1}~and~\eqref{eqn:constructf2}.
For any $y\in\mathbb{R}^m$, Borel set $\mathcal{A}\subseteq \mathbb R^n$, and time $k$, 
\begin{align*}
  &\mathbb{P}_{f'} (\hat{x}_k \in\mathcal{A} | \avg({\bm Y}(k)) = y ) \\
  =&  \int_{ Y\in\mathbb{R}^{m\times k}}  \mathbb{P}_f( \hat{x}_k \in\mathcal{A}   |{\bm Y}(k) = Y)\\
  &\qquad \qquad {\rm d}\mathbb{P}({\bm Z}(k) = Y | \avg({\bm Z}(k) ) = y).\addtag \label{eqn:constructf1}  	
\end{align*}
The above equation is the integral of $\mathbb{P}_f(\cdot)$ over $Y\in\R^{m\times k}$ with respect to the  conditional probability measure $\mathbb{P}({\bm Z}(k)  | \avg({\bm Z}(k) ) = y)$.
Notice that this conditional probability measure $\mathbb{P}({\bm Z}(k)  | \avg({\bm Z}(k) ) = y)$ is well-defined since $\avg({\bm Z}(k) )$ is a sufficient statistic of the ``true'' measurements ${\bm Z}(k)$ for the underlying state $x$, i.e., for any state $x$,
\[\mathbb{P}_x({\bm Z}(k)  | \avg({\bm Z}(k) ) = y) = \mathbb{P}({\bm Z}(k)  | \avg({\bm Z}(k) ) = y),\]
where $\mathbb{P}_x(\cdot)$ denotes the probability measure governing the original measurements ${\bm Z}(k)$ when the state $x$ is given.
Notice that RHS of~\eqref{eqn:constructf1} can be interpreted as ``taking expectation'' of the conditional probability measure $\mathbb{P}_f( \hat{x}_k|{\bm Y}(k))$ given that $\avg({\bm Y}(k)) = y$ and that ${\bm Y}(k)$ shares the same distribution with ${\bm Z}(k)$.

Furthermore, let $f'$ be in $\mathcal{F}_{\rm c}$, i.e.,
\begin{align} \label{eqn:constructf2}
  \mathbb{P}_{f'}(\hat{x}_k\in\mathcal{A}|{\bm Y}(k)) = \mathbb{P}_{f'}(\hat{x}_k\in\mathcal{A}|{\bm Y}'(k))
\end{align}
for any Borel set $\mathcal{A}\subseteq\R^n$ whenever $\avg({\bm Y}(k)) = \avg({\bm Y}'(k))$.

Let  $\mathcal{B}_\delta(x)$ denote the closed ball centered at $x\in\R^n$ with radius $\delta >0$:
\begin{align} \label{eqn:closedball}
	\mathcal{B}_\delta(x) \triangleq \{y\in\R^n: \|  y -x \|_2 \leq \delta  \}.
\end{align}
 Regarding with $f$ and $f'$, in the remainder of this proof we devote ourselves to  showing that the following inequality holds for  any state $x$, set $\mathcal{C}$, $\delta>0$ and time $k$:
\begin{align*} 
  \sup_{g\in\mathcal{G}}\mathbb{P}_{f,g,x,\mathcal{C}}\left(  \hat x_k \not\in\mathcal{B}_\delta(x)    \right) 
  \geq \sup_{g\in\mathcal{G}} \mathbb{P}_{f',g,x,\mathcal{C}}\left( \hat x_k \not\in\mathcal{B}_\delta(x)  \right), \addtag \label{eqn:probequality}
\end{align*}
from which Lemma~\ref{lemma:compressed_estimator} follows straightforwardly. 

We first identify the most harmful attack strategy for a generic $f$. Given state $x$, set $\mathcal{C}$, $\delta>0$, time $k$, and estimator $f$, consider the following optimization problem:
\begin{align*}
 &\max_{Y_2\in\R^{q\times k}} \int_{ Y_1\in\mathbb{R}^{(m-q)\times k}} \mathbb{P}_f\left( \hat{x}_k \not\in\mathcal{B}_\delta(x)  \Big|{\bm Y}(k)_{\mathcal{M}\setminus \mathcal{C}} = Y_1, \right. \\
	&\qquad \qquad \qquad \quad {\bm Y}(k)_{\mathcal{C}} = Y_2 \Big) {\rm d}\mathbb{P}_x({\bm Z}(k)_{\mathcal{M}\setminus\mathcal{C}} = Y_1). \addtag \label{eqn:optCompress}
\end{align*}
Denote its optimal solution (i.e., the ``manipulated matrix") as 
$\mm(f,x,\mathcal{C},\delta,k)$. Then
one may see that changing the measurements of the compromised sensors available at time $k$, ${\bm Y}(k)_{\mathcal{C}}$, to $\mm(f,x,\mathcal{C},\delta,k)$ would maximize the error\footnote{For the sake of presentation, we call the event $\hat{x}_k \not\in\mathcal{B}_\delta(x)$ an error.  } probability under estimator~$f$. The optimal value of the optimization problem~\eqref{eqn:optCompress} is just the worst-case error probability $\sup_{g\in\mathcal{G}}\mathbb{P}_{f,g,x,\mathcal{C}}\left( \hat x_k \not\in\mathcal{B}_\delta(x)   \right)$.

We then identify the most harmful attack strategy for the compressed estimator $f'$. Given state $x$, set $\mathcal{C}$, $\delta>0$, time $k$, and estimator $f'$, consider the following optimization problem:
\begin{align*}
& \max_{y_2\in\R^{q}} \int_{y_1\in \mathbb{R}^{m-q}} \mathbb{P}_{f'}\left( \hat{x}_k \not\in\mathcal{B}_\delta(x)  \Big|\avg({\bm Y}(k))_{\mathcal{M}\setminus \mathcal{C}} = y_1, \right. \\
&\quad \avg({\bm Y}(k))_{\mathcal{C}} = y_2 \Big) {\rm d}\mathbb{P}_x(\avg({\bm Z}(k))_{\mathcal{M}\setminus\mathcal{C}} = y_1).  \addtag \label{eqn:optCompress1}
\end{align*}
Denote its optimal solution (i.e., the ``manipulated vector") as 
$\mv(f',x,\mathcal{C},\delta,k)$.
One may verify that changing the measurements of the compromised sensors available at time $k$ such that $\avg({\bm Z}(k))_{\mathcal{C}} =\mv(f',x,\mathcal{C},\delta,k) $  would maximize the error probability under estimator~$f'$. The optimal value of the optimization problem~\eqref{eqn:optCompress1} is just the worst-case error probability $\sup_{g\in\mathcal{G}}\mathbb{P}_{f',g,x,\mathcal{C}}\left(\hat x_k \not\in\mathcal{B}_\delta(x)    \right)$.
For the sake of better presentation, in the remainder of this proof, for any matrix ${\bm M}$, we rewrite ${\bm M}_{\mathcal{M}\setminus \mathcal{C}}$ as ${\bm M}_{[1]}$ and ${\bm M}_{\mathcal{C}}$ as ${\bm M}_{[2]}$. We also omit the time index $k$ of ${\bm Z}(k)$ and ${\bm Y}(k)$. The set $\mathcal{B}_\delta(x) $ is denoted by $\mathcal{B}$.
Notice that the ``true'' measurements ${\bm Z}$ are independent across sensors given the underlying state $x$. Therefore, we can rewrite~\eqref{eqn:constructf1} as follows:
\begin{align*}
  &\mathbb{P}_{f'} (\hat{x}_k \in\mathcal{A} | \avg({\bm Y}) = y ) \\
  =&  \int_{ \mathbb{R}^{q\times k}} \int_{ \mathbb{R}^{(m-q)\times k}}   \mathbb{P}_f( \hat{x}_k \in\mathcal{A}   |{\bm Y}_{[1]} = Y_{[1]}, {\bm Y}_{[2]} = Y_{[2]} )\\
  &\qquad\qquad  {\rm d}\mathbb{P}({\bm Z}_{[1]} = Y_{[1]} | \avg({\bm Z}_{[1]}) = y_{[1]})\\
  &\qquad\qquad\quad {\rm d}\mathbb{P}({\bm Z}_{[2]} = Y_{[2]} | \avg({\bm Z}_{[2]}) = y_{[2]}). 
\end{align*}
Then one obtains that
\begin{align*}
  &\sup_{g\in\mathcal{G}}\mathbb{P}_{f',g,x,\mathcal{C}}\left(\hat x_k \not\in\mathcal{B}_\delta(x)    \right) \\
    =& \int_{ \mathbb{R}^{m-q}} \int_{ \mathbb{R}^{(m-q)\times k}} \int_{ \mathbb{R}^{q\times k}}  \mathbb{P}_f( \hat{x}_k \not\in\mathcal{B}   |{\bm Y}_{[1]} = Y_{[1]}, {\bm Y}_{[2]} = Y_{[2]} )\\
  &\qquad\qquad{\rm d}\mathbb{P}({\bm Z}_{[2]} = Y_{[2]} | \avg({\bm Z}_{[2]}) = \mv(f',x,\mathcal{C},\delta,k))  \\
  &\qquad\qquad\quad {\rm d}\mathbb{P}({\bm Z}_{[1]} = Y_{[1]} | \avg({\bm Z}_{[1]}) = z_{[1]}) \\
  &\qquad\qquad\qquad {\rm d}\mathbb{P}_x(\avg({\bm Z}_{[1]}) = z_{[1]})\\
  =&  \int_{ \mathbb{R}^{(m-q)\times k}} \int_{ \mathbb{R}^{q\times k}}  \mathbb{P}_f( \hat{x}_k \not\in\mathcal{B}   |{\bm Y}_{[1]} = Y_{[1]}, {\bm Y}_{[2]} = Y_{[2]} )\\
  &\qquad\qquad{\rm d}\mathbb{P}({\bm Z}_{[2]} = Y_{[2]} | \avg({\bm Z}_{[2]}) = \mv(f',x,\mathcal{C},\delta,k))  \\
  &\qquad\qquad\quad {\rm d}\mathbb{P}_x({\bm Z}_{[1]} = Y_{[1]} ) \\
    =& \int_{ \mathbb{R}^{q\times k}} \int_{ \mathbb{R}^{(m-q)\times k}}   \mathbb{P}_f( \hat{x}_k \not\in\mathcal{B}   |{\bm Y}_{[1]} = Y_{[1]}, {\bm Y}_{[2]} = Y_{[2]} )\\
  &\qquad\qquad{\rm d}\mathbb{P}_x({\bm Z}_{[1]} = Y_{[1]} )  \\
  &\qquad\qquad {\rm d}\mathbb{P}({\bm Z}_{[2]} = Y_{[2]} | \avg({\bm Z}_{[2]}) = \mv(f',x,\mathcal{C},\delta,k))\\ 
  \leq& \max_{Y_{[2]}\in\R^{q\times k}} \int_{ \mathbb{R}^{(m-q)\times k}}   \mathbb{P}_f( \hat{x}_k \not\in\mathcal{B}   |{\bm Y}_{[1]} = Y_{[1]}, {\bm Y}_{[2]} = Y_{[2]} )\\
  &\qquad\qquad{\rm d}\mathbb{P}_x({\bm Z}_{[1]} = Y_{[1]} )  \\
  &=\sup_{g\in\mathcal{G}}\mathbb{P}_{f,g,x,\mathcal{C}}\left( \hat x_k \not\in\mathcal{B}_\delta(x)   \right),
\end{align*}
where the second equality
follows from the law of total probability, and the inequality holds because $\mathbb{P}({\bm Z}_{[2]} | \avg({\bm Z}_{[2]}) = y_{[2]})$ is a probability measure for any $y_{[2]}$, i.e.,
\[\int_{ \mathbb{R}^{q\times k}} {\rm d}\mathbb{P}({\bm Z}_{[2]} = Y_{[2]} | \avg({\bm Z}_{[2]}) = y_{[2]}) =1 \]
for any $y_{[2]}$.
The proof is thus complete. 

In order for readers to have a better picture of relationship between the main results obtained in this paper (e.g., Theorems~\ref{theorem:restrict_estimator}~and~\ref{theorem:optimalestimator}) and the assumptions posed in Section~\ref{sec:problem}, in particular, Assumptions~\ref{assumpt:sparseattack}~and~\ref{assumpt:knowledge}, in the following remark, we explain in detail how these assumptions are utilized to derive~\eqref{eqn:optCompress}. 
\begin{remark} \label{remark:relation}
	Due to Assumption~\ref{assumpt:sparseattack}, one could split the sensors into two groups: ``good" ones in $\mathcal M \setminus \mathcal C$ and ``bad" ones in $\mathcal{C}$. All the measurements from up to time $k$ from good sensors are not manipulated and, thus, denoted by $\bm Z(k)_{\mathcal M \setminus \mathcal C}$. The attacker can develop the term $\mathbb{P}_x({\bm Z}(k)_{\mathcal{M}\setminus\mathcal{C}} = Y_1)$  since it knows the true state $x$ and system parameters by Assumption~\ref{assumpt:knowledge}. Furthermore, since the attacker knows the estimator $f$ and has unlimited memory, it is proper to obtain the most harmful attack strategy by~\eqref{eqn:optCompress}.
\end{remark}

\section{Proof of Lemma~\ref{lemma:cd_estimator}}
\label{appendix:cd_estimator}
\subsection{Preliminaries}
In the following lemma, we bound the area where a random variable has a high probability showing up.

Given any random variable $y\in\R^n$, we shall say that a point $x$ is $\delta$-typical, if
\begin{align}
\mathbb{P}(y\in\mathcal{B}_\delta(x))> n/(n+1),
\label{eqn:deftypical}
\end{align}
where $\mathcal{B}_\delta(x)$ is the closed ball defined in~\eqref{eqn:closedball}.
In other words, $y$ has a high probability lying in the $\delta$-neighborhood of $x$.

Let $\mathcal{T}_\delta(y)$ denote the set of all $\delta$-typical point $x$ w.r.t. a random variable $y$, i.e.,
\begin{align} \label{eqn:defA}
\mathcal{T}_\delta(y) \triangleq \{x\in\R^n:x\text{ is $\delta$-typical w.r.t. $y$}\}.
\end{align}
We have the following lemma to show that $\mathcal{T}_\delta(y)$ lies in a ball with radius $\delta$:

\begin{lemma} \label{lemma:probgeometric}
	For any random variable $y\in\R^n$, there exists $x^*\in\R^n$ such that 
	\begin{align} \label{eqn:smallerthanball}
	\mathcal{T}_\delta(y) \subseteq \mathcal{B}_\delta(x^*).
	\end{align}	
\end{lemma} 
To proceed, we need the following lemma:
\begin{lemma} \label{lemma:probinter}
	Let $\mathcal{A}_1, \ldots, \mathcal{A}_n$ be $n$ random events with the same underlying probability space, then it holds that
	\begin{align} \label{probinter}
	\mathbb{P}(\cap_{j=1}^n\mathcal{A}_j) \geq \sum_{j=1}^{n} \mathbb{P}(\mathcal{A}_j) -n+1.
	\end{align}
\end{lemma}
\begin{proof}[Proof of Lemma~\ref{lemma:probinter}]
	\begin{align*}
	\mathbb{P}(\cap_{j=1}^n\mathcal{A}_j) &= 1 - 	\mathbb{P}(\cup_{j=1}^n\mathcal{A}_j^c)  \geq 1 - \sum_{j=1}^n \mathbb{P}(\mathcal{A}_j^c)\\
	& = 1-\sum_{j=1}^n\left(1-\mathbb{P}(\mathcal{A}_j)\right) = \sum_{j=1}^{n} \mathbb{P}(\mathcal{A}_j) -n+1,
	\end{align*}
	where $\mathcal A^c$ is the complement of set $\mathcal A$. The proof is thus complete.
\end{proof}

\begin{proof}[Proof of Lemma~\ref{lemma:probgeometric}]
If $\mathcal{T}_\delta(y)$ is empty, then obviously $\mathcal{T}_\delta(y) = \emptyset \subset \mathcal B_\delta(0)$.

If $\mathcal{T}_\delta(y)$ only contains $j\leq (n+1)$ elements, say, $x_1, \ldots, x_{j}$. Then Lemma~\ref{lemma:probinter} together with~\eqref{eqn:deftypical} yields that 
\[\mathbb{P}(y\in\cap_{i=1}^j\mathcal{B}_\delta(x_i)) > j\times\frac{n}{n+1}-j+1\geq 0,\]
which means that the set $\cap_{i=1}^j\mathcal{B}_\delta(x_i)$ is not empty. Then $\mathcal{T}_\delta(y)\subseteq \mathcal{B}_\delta(x^*)$ for some $x^*\in\cap_{i=1}^j\mathcal{B}_\delta(x_i)$.

If $\mathcal{T}_\delta(y)$ contains $j>(n+1)$ elements ($j$ might be infinite). Then again by Lemma~\ref{lemma:probinter}, one obtains that  for  \emph{any} $n+1$ elements, say,  $x_1, \ldots, x_{n+1}$, there holds
\[\mathbb{P}(y\in\cap_{i=1}^{n+1}\mathcal{B}_\delta(x_i)) > 0,\] 
that is, $\cap_{i=1}^{n+1}\mathcal{B}_\delta(x_i) \not = \emptyset$. Since $\mathcal{B}_\delta(x)$ is compact and convex for any $x$, then Helly's theorem~\cite{danzer1963helly} means that 
\[\cap_{x\in\mathcal{T}_\delta(y)} \mathcal{B}_\delta(x) \not =\emptyset. \] Then  $\mathcal{T}_\delta(y)\subseteq \mathcal{B}_\delta(x^*)$ for some $x^*\in\cap_{x\in\mathcal{T}_\delta(y)}\mathcal{B}_\delta(x)$. The proof is thus complete.
\end{proof}
\begin{definition}
	Any point $x^*\in\R^n$ is said to be a $\delta$-center of a random variable $y\in\R^n$ if it is such that~\eqref{eqn:smallerthanball} holds. 
\end{definition}

The following follows readily from Lemma~\ref{lemma:probgeometric}.
\begin{lemma} \label{lemma:indicatorlower}
	If $x^*\in\R^n$ is a $\delta$-center of a random variable $y\in\R^n$, then for any $x\in\R^n$:
	\begin{align}
		1- \mathbbm{1}_{\mathcal B_\delta(x)}(x^*) \leq (n+1) \mathbb{P}(y\notin\mathcal B_\delta(x)).
	\end{align}
\end{lemma}

\subsection{Main Body}
%

Consider any estimator $f \in \mathcal F$, we construct a deterministic one $f'\in\mathcal{F}_{\rm d}$: for any time $k$ and observations ${\bm Y}(k)$, let $f'_k({\bm Y}(k))$ be a $\delta$-center of the random variable $f_k({\bm Y}(k))$, the existence of which is guaranteed by Lemma~\ref{lemma:probgeometric}.



Hence, for any attack strategy $g$, true state $x$, set of compromised sensors $\mathcal{C}$, and time $k$, we have
\begin{align*}
	&\mathbb{P}_{f',g,x,\mathcal{C}}(\|\hat{x}_k-x\| > \delta)\\
	=& \int_{Y\in\R^{m\times k}} 	1- \mathbbm{1}_{\mathcal B_\delta(x)}(f'_k(Y)) \,{\rm d}\mathbb{P}_{g,x,\mathcal{C}}({\bm Y}(k)=Y)\\
	\leq& \int_{Y\in\R^{m\times k}} 	(n+1) \mathbb{P}(f_k(Y)\notin\mathcal B_\delta(x))  \,{\rm d}\mathbb{P}_{g,x,\mathcal{C}}({\bm Y}(k)=Y) \\
	=& (n+1) \mathbb{P}_{f,g,x,\mathcal{C}}(\|\hat x(k)-x\| > \delta),
\end{align*}
where the inequality follows from Lemma~\ref{lemma:indicatorlower}, and  $\mathbb{P}_{g,x,\mathcal{C}}(\cdot)$ denotes the probability measure governing ${\bm Y}(k)$ when $g,x,\mathcal{C}$ are given.
Then it is clear that
\begin{align} \label{eqn:worstcasebound}
	e(f,k,\delta) \geq  e(f',k,\delta)/(n+1).
\end{align}
Recall that $e(f,k,\delta)$ is the worst-case probability defined in~\eqref{eqn:errorProb}. Then it follows that for any $\delta>0$:
\begin{align*}
	r(f,\delta) =& \liminf_{k\rightarrow\infty} - \frac{\log e(f,k,\delta)}{k} \\
	\leq & \liminf_{k\rightarrow\infty} - \frac{\log e(f',k,\delta)/(n+1)}{k} \\
	=& \liminf_{k\rightarrow\infty} - \frac{\log e(f',k,\delta)}{k} \addtag \label{eqn:asympequal}  \\
	=& r(f',\delta).
\end{align*}
The proof is thus complete. 

\section{Proof of Theorem~\ref{theorem:restrict_estimator}}
\label{appendix:theoremRestEst}
Consider a compressed but possibly random estimator $f \in \mathcal F_{\rm c}$, we construct a deterministic one $f'\in\mathcal{F}_{\rm cd}$ satisfying:
\begin{itemize}
	\item for any time $k$ and observations ${\bm Y}(k)$, $f'({\bm Y}(k))$ is a $\delta$-center of the random variable $f({\bm Y}(k))$;
	\item $f'({\bm Y}(k)) = f'({\bm Y'}(k))$ if $\avg({\bm Y}(k)) = \avg({\bm Y}'(k))$.
\end{itemize}
The existence of $f'$ is guaranteed by Lemma~\ref{lemma:probgeometric} and the fact that random variables $f({\bm Y}(k)), f({\bm Y'}(k))$ have the same distribution if $\avg({\bm Y}(k)) = \avg({\bm Y}'(k))$.

Then as in Appendix~\ref{appendix:cd_estimator}, one obtains $r(f,\delta) \leq r(f',\delta)$, which, together with Lemma~\ref{lemma:compressed_estimator}, concludes the proof.


\section{Proof of Theorem~\ref{theorem:optimalestimator}}
\label{proof:opt-estimator}
We first prove that $r(f,\delta)$ is upper-bounded by $u(\delta)$ for any $f\in\mathcal{F}$ in Lemma~\ref{lemma:upperbound}. We then show that $r(f^*_{\delta},\delta) = u(\delta)$ in Lemma~\ref{lemma:achieve}. Before proceeding, we need the following supporting definitions and lemmas.

\subsection{Supporting Definition and Lemmas}
\begin{lemma} \label{lemma:radius}
  For any $\mathcal{A}\subseteq\R^n$, if $\rad(\mathcal{A}) > \gamma$, there exists a set $\mathcal{A}_0\subseteq\mathcal{A}$ such that
  $|\mathcal{A}_0| \leq n+1$ and
  $\rad(\mathcal{A}_0) > \gamma$.
\end{lemma}
\begin{proof}
  If $|\mathcal{A}|\leq n+1$, then the lemma holds trivially by letting $\mathcal{A}_0 = \mathcal{A}$. When $|\mathcal{A}| > n+1$, we prove the lemma by contradiction. Suppose for every subset $\mathcal{A}_0\subseteq \mathcal{A}$ with $|\mathcal{A}_0|\leq n+1$, there holds $\rad(\mathcal{A}_0) \leq \gamma$. Then we have
  \begin{align*}
    \cap_{x\in\mathcal{A}_0} \mathcal{B}_{\gamma}(x) \not= \emptyset
  \end{align*}
  for every $\mathcal{A}_0\subseteq \mathcal{A}$ with $|\mathcal{A}_0| = n+1$. Since $\mathcal{B}_{\gamma}(x)$ is compact and convex for every $x\in\mathcal{A}$, Helly's theorem~\cite{danzer1963helly} implies that
  \begin{align*}
    \cap_{x\in\mathcal{A}} \mathcal{B}_{\gamma}(x) \not= \emptyset.
  \end{align*}
  Hence, for any $x_0\in \cap_{x\in\mathcal{A}} \mathcal{B}_{\gamma}(x)$, $\mathcal A \subseteq \mathcal B_\gamma(x_0)$. Therefore $\rad(\mathcal{A}) \leq \gamma$, which contradicts the condition $\rad(\mathcal{A}) > \gamma$.
\end{proof}

With a slight abuse of notation, we use the sequence of functions $(f_1(\avg({\bf Y}(1))), f_2(\avg({\bf Y}(2))),\ldots )$ from time $1$ to $\infty$ to denote a compressed and deterministic estimator $f\in\mathcal{F}_{\rm cd}$.
\begin{definition} \label{def:set}
  Given a compressed and deterministic estimator $f\in\mathcal{F}_{\rm cd}$, $x\in\R^n$, $\delta>0$, and time $k$, 	let $\mathcal{Y}(f,x,\delta,k)$ be the set of averaged measurements $\avg({\bf Y}(k))$ such that the estimate $\hat{x}_k$ lies outside the ball $\mathcal{B}_\delta(x)$, i.e.,
  \begin{align}
    \mathcal{Y}(f,x,\delta,k) \triangleq \{y\in\R^m: f_k(y) \not\in \mathcal{B}_\delta(x) \}.
  \end{align}
\end{definition}
%

\subsection{Upper Bound}
\begin{lemma} \label{lemma:upperbound}
  For any estimator $f\in\mathcal{F}$, there holds
  \begin{align}
    r(f,\delta) \leq u(\delta).
  \end{align}
\end{lemma}
\begin{proof}
  We show that $r(f,\delta) < u(\delta) +\epsilon$ for any $\epsilon>0$ and $f\in\mathcal{F}$. 
  
  Given $\epsilon>0$, from the definition of $u(\delta)$, one obtains that there must exist $y^*\in\R^m$ and a set $\mathcal{A}\subseteq \R^n$  such that:
  \begin{enumerate}
  \item $d_x(y^*)\leq u(\delta) + \epsilon/2$ for all $x\in\mathcal{A}$;
  \item $\rad(\mathcal{A}) > \delta $.
  \end{enumerate}	
  Notice that the above $y^*$ and $\mathcal{A}$ can be constructed as follows. By the definition of $u(\delta)$, there must exist a $y^*$ such that for every $x\in\mathbb X(y^*,\delta)$, $d_x(y^*)< u(\delta) + \epsilon/4$ holds. Then we construct $\mathcal{A}$ by cases. If $\rad(\mathbb X(y^*,\delta)) = \delta$, then there must exist $\phi^*< u(\delta) + \epsilon/4$ such that $\mathcal X(y^*,\phi^*) =  \mathbb X(y^*,\delta)$. Let $\mathcal{A} = \mathcal X(y^*,\phi^*+\epsilon/4)$, and, therefore, $\mathcal{A} \subset \mathcal X(y^*, u(\delta) + \epsilon/2)$. Also, by the third bullet of Lemma~\ref{lemma:continuous}, $\rad(\mathcal{A}) > \delta $ holds. If $\rad(\mathbb X(y^*,\delta)) < \delta$, then let $\phi^* = \min\{\phi: \mathbb X(y^*,\delta) \subseteq \mathcal X(y^*,\phi) \}$ and $\mathcal{A} = \mathcal X(y^*,\phi^*)$. Then by  Lemma~\ref{lemma:continuous},   $\rad(\mathcal{A}) > \delta $ and $d_x(y^*)\leq u(\delta) + \epsilon/4$ for all $x\in\mathcal{A}$.
  
  Lemma~\ref{lemma:radius} yields that there exists $\mathcal{A}_0\subseteq \mathcal{A}$ such that $\rad(\mathcal{A}_0)>\delta$ and $|\mathcal{A}_0| \leq n+1$. Let $a^*(y,x)$ be the optimal solution to the optimization problem in~\eqref{eqn:opt1} given $y\in\R^m$ and $x\in\R^n$.  Since $d_x(y,\mathcal{I})$ in~\eqref{eqn:distancerestict} is continuous w.r.t. $y$ and $|\mathcal{A}_0| \leq n+1$, then one obtains that 
  there exists a ball $\mathcal{B}_\beta(y^*)$, where $\beta>0$ is dependent on $\epsilon$, such that $d_x(y,\mathcal{M}\setminus\supp(a^*(y^*,x)))<u(\delta) +\epsilon$ for every $x\in\mathcal{A}_0$ and every $y\in\mathcal{B}_\beta(y^*)$.

  By Theorem~\ref{theorem:restrict_estimator}, one suffices to  consider a compressed and deterministic estimator $f\in\mathcal{F}_{\rm cd}$. Furthermore, since $\rad(\mathcal{A}_0) > \delta$,   one concludes that for every time $k$ and $f\in\mathcal{F}_{\rm cd}$, there holds
  \begin{align}
    \mathcal{B}_\beta(y^*) \subseteq \cup_{x\in\mathcal{A}_0} \mathcal{Y}(f,x,\delta,k).
  \end{align}
  Let $\mathfrak{L}_n(\cdot)$ denote the Lebesgue measure on $\R^n$. 
  Because of countable additivity of Lebesgue measure~\cite{rudin1964principles},  one obtains that there must exist a point $x^*\in\mathcal{A}_0$ such that
  \begin{align}   \label{eqn:lebesgueall}
    \mathfrak{L}_m( \mathcal{B}_\beta(y^*) \cap \mathcal{Y}(f,x^*,\delta,k) )   \geq  \mathfrak{L}_m(\mathcal{B}_\beta(y^*))/(n+1).
  \end{align}
  For the sake of simplicity, let $\mathcal{B}(x^*, k)  \triangleq \mathcal{B}_\beta(y^*) \cap \mathcal{Y}(f,x^*,\delta,k)  $ and $ \mathcal{I}^* \triangleq \supp(a^*(y^*,x^*))$. Then it is clear that
  \begin{align*}
    &e(f,k,\delta) \\
    \geq& \sup_{g\in\mathcal{G}} \mathbb{P}_{g,x^*,\mathcal{I}^*}\left(\avg({\bf Y}(k)) \in \mathcal{Y}(f,x^*,\delta,k)\right) \\
    \geq& \sup_{g\in\mathcal{G}} \mathbb{P}_{g,x^*,\mathcal{I}^*}\left(\avg({\bf Y}(k)) \in \mathcal{B}(x^*, k) \right) \\
    =&\sup_{\avg({\bf Z}(k))_{\mathcal{I^*}}\in\R^q}  \mathbb{P}_{x^*}\left( \avg({\bf Z}(k))\in\mathcal{B}(x^*, k)   \right)\\
    =& \sup_{o\in\R^q}  \mathbb{P}_{x^*}\left( \avg({\bf Z}(k))_{\mathcal{M}\setminus \mathcal{I}^*} \in	\B(o,k)   \right),
     \addtag \label{eqn:prob}
  \end{align*}
  where $\B(o,k)$ with $o\in\R^q$ is the projected set of $\mathcal{B}(x^*, k)$:
  \begin{align*}
  	\B(o,k) \triangleq \{y_{\mathcal{M}\setminus \mathcal{I}^*}: y\in\mathcal{B}(x^*, k), y_{ \mathcal{I}^*} = o\}.
  \end{align*}
  Further let $\mathbb{B}(k)$ be a set that satisfies:
  \begin{align*}
  	\mathfrak{L}_{m-q}( \mathbb{B}(k)) = \sup_{o\in\R^q} \mathfrak{L}_{m-q}( \mathbb{B}(o,k)), 
  \end{align*}
  and for any $\upsilon>0$, there exists $o\in\R^q$ such that 
  \[\mathfrak{L}_{m-q}( \mathbb{B}(k) \setminus \mathbb{B}(o,k)  ) <\upsilon. \]
  Roughly speaking, $\mathbb{B}(k)$ can be viewed as the supremum set. Then one obtains that 
  \begin{align*}
  	&\sup_{o\in\R^q}  \mathbb{P}_{x^*}\left( \avg({\bf Z}(k))_{\mathcal{M}\setminus \mathcal{I}^*} \in	\B(o,k)   \right) \\
  	&\geq \mathbb{P}_{x^*}\left( \avg({\bf Z}(k))_{\mathcal{M}\setminus \mathcal{I}^*} \in	\B(k)   \right) \addtag \label{eqn:prob1}
  \end{align*}
  In the following, we focus on characterizing the term in~\eqref{eqn:prob1}.  Let $p_{x^*}(\cdot): \R^{m-q} \mapsto \R_+$ be the probability density of $\avg({\bf Z}(k))_{\mathcal{M}\setminus \mathcal{I}^*}$ conditioned on the underlying state $x^*$, i.e.,
  \begin{align*}
    p_{x^*}(z) = \mathcal{N}({\bm H}_{\mathcal{M}\setminus \mathcal{I}^*}, {\bm W}_{\{\mathcal{M}\setminus \mathcal{I}^*\}}/k, z),
  \end{align*} 
  where ${\bm W} = \diag(W_i,W_2,\ldots,W_m)$ is the diagonal matrix, ${\bm W}_{\{\mathcal{M}\setminus \mathcal{I}^*\}}$ (different from ${\bm W}_{\mathcal{M}\setminus \mathcal{I}^*}$) the square matrix obtained from ${\bm W}$ after removing all of the rows and columns except those in the index set $\mathcal{M}\setminus \mathcal{I}^*$, and $\mathcal{N}({\bm \mu},{\bm \Sigma}, {\bm x})$ the probability density function of a Gaussian random variable with mean ${\bm \mu}$ and variance ${\bm \Sigma}$ taking value at ${\bm x}$. Then one obtains that
  \begin{align*}
    &	\mathbb{P}_{x^*}\left(\avg({\bf Z}(k))_{\mathcal{M}\setminus \mathcal{I}^*} \in \B(k) \right)\\
    =& \int_{\R^{m-q}} \mathbbm{1}_{\B(k)}(z) p_{x^*}(z) {\rm d} z.
  \end{align*}
  From~\eqref{eqn:lebesgueall}, some basic arguments mainly involving the regularity theorem for Lebesgue measure and the Heine--Borel theorem~\cite{rudin1964principles} give that there exists $\gamma>0$ such that
  \begin{align}   \label{eqn:lebesgue}
    \mathfrak{L}_{m-q} ( \B(k)  ) > \gamma \mathfrak{L}_{m-q} ( \mathcal{B}_\beta(y^*)_{\mathcal{M}\setminus \mathcal{I}^*}   ).
  \end{align}
  Furthermore, $\gamma$ can be  determined by $m,n,q$, and $\beta$, which is, in particular, irrelevant to time $k$. Let $\mathcal{Z}(x^*,k) \subseteq \mathcal{B}_\beta(y^*)_{\mathcal{M}\setminus \mathcal{I}^*}  $ be the pre-image of $(\underline{p},\bar{p})$ under the function $p_{x^*}(\cdot)$, where $\underline{p}\triangleq \min_{z\in\mathcal{B}_\beta(y^*)_{\mathcal{M}\setminus \mathcal{I}^*}} p_{x^*}(z)$ is the minimum value\footnote{Notice that this minimum can be attained since $p_{x^*}(z)$ is a continuous function and $\mathcal{B}_\beta(y^*)_{\mathcal{M}\setminus \mathcal{I}^*}$ is compact.} and $\bar{p}$ is such that  
  \begin{align}
    \mathfrak{L}_{m-q} ( \mathcal{Z}(x^*,k)  ) = \gamma \mathfrak{L}_{m-q} ( \mathcal{B}_\beta(y^*)_{\mathcal{M}\setminus \mathcal{I}^*}   ).
  \end{align}
  Notice that $\bar{p}$ exists since $	\mathfrak{L}_{m-q} ( \{z:  p_{x^*}(z) =p  \}  ) = 0  $ for any $p$. Then one obtains that
  \begin{align*}
    &	\mathbb{P}_{x^*}\left(\avg({\bf Z}(k))_{\mathcal{M}\setminus \mathcal{I}^*} \in \B(k) \right)\\
    \geq& \mathbb{P}_{x^*}\left(\avg({\bf Z}(k))_{\mathcal{M}\setminus \mathcal{I}^*} \in \mathcal{Z}(x^*,k) \right). \addtag 
          \label{eqn:probless}
  \end{align*}
  Notice that the pre-image of an open set under a continuous function is also open, $\mathcal{Z}(x^*,k)$ is thus open.  Furthermore, since both $\gamma$ and $\mathcal{B}_{\beta}(y^*)$ are independent of time $k$, $\mathcal{Z}(x^*,k)$ will be an nonempty set whatever $k$ is.  Therefore, the following holds:
  \begin{align*}
    &\limsup_{k\to\infty} \frac{1}{k}	\log \mathbb{P}_{x^*}\left(\avg({\bf Z}(k))_{\mathcal{M}\setminus \mathcal{I}^*} \in \mathcal{Z}(x^*,k) \right) \\
    \leq & \inf_{z\in\mathcal{Z}(x^*,k)} \frac{1}{2}(z- {\bm H}_{\mathcal{M}\setminus \mathcal{I}^*} x^* )^\top  \left({\bm W}_{\mathcal{M}\setminus \mathcal{I}^*}\right)^{-1} (z- {\bm H}_{\mathcal{M}\setminus \mathcal{I}^*} x^* )\\
    = & \inf_{z\in\R^m, z_{\mathcal{M}\setminus \mathcal{I}^*}\in\mathcal{Z}(x^*,k)} \, d_{x^*}(y,\mathcal{M}\setminus\mathcal{I}^*) \\
    <& u(\delta) + \epsilon, \addtag  \label{eqn:cramer1}
  \end{align*}
  where first inequality is due to the Cram\'{e}r's Theorem~\cite{dembo2009large} and the fact that $\mathcal{Z}(x^*,k)$ is open and $d_{x^*}(\cdot,\mathcal{M}\setminus\mathcal{I}^*)$ is the corresponding rate function since
  the observation noise $w_i(k)$ is i.i.d. across time and independent across the sensors; the last inequality follows from the definitions of $\mathcal{Z}(x^*,k)$ and $\mathcal{B}_{\beta}(y^*)$. Then, combining with~\eqref{eqn:prob},~\eqref{eqn:prob1}~and~\eqref{eqn:probless}, one concludes the proof.
\end{proof}

\subsection{Achievability}
About the estimator $f^*_{\delta}$ defined in~\eqref{eqn:optestimatork}, we have the following lemma:
\begin{lemma}  \label{lemma:achieve}
  There holds $r(f^*_{\delta},\delta) = u(\delta)$.
\end{lemma}
\begin{proof}
  Notice that, by the definition of $u(\delta)$, for any $x$, $\delta$ and $k$, if $y\in\mathcal{Y}(f^*_{\delta},x,\delta,k)$, then $d_x(y) \geq u(\delta)$. Recall that $\mathcal{Y}(\cdot,\cdot,\cdot,\cdot)$ is introduced in Definition~\ref{def:set}.
  Let 
  \begin{align*}
  \mathcal{Y}^*(x) \triangleq \{y: d_x(y) \geq u(\delta)\}.
  \end{align*}
  Then $\mathcal{Y}(f^*_{\delta},x,\delta,k) \subseteq \mathcal{Y}^*(x)$ holds. Therefore, for any $k$, $x$ and $\mathcal{I}$: 
  \begin{align*}
    & \sup_{g\in\mathcal{G}} \mathbb{P}_{g,x,\mathcal{I}}\left(\avg({\bf Y}(k)) \in \mathcal{Y}(f^*_{\delta},x,\delta,k)\right) \\
    \leq& \sup_{g\in\mathcal{G}} \mathbb{P}_{g,x,\mathcal{I}}\left(\avg({\bf Y}(k)) \in \mathcal{Y}^*(x)\right) \\
    \leq&  \mathbb{P}_{x}\left(\avg({\bf Z}(k))_{\mathcal{M}\setminus \mathcal{I}} \in \mathcal{Y}^*(x)_{\mathcal{M}\setminus \mathcal{I}} \right).
  \end{align*}
 Then similar to~\eqref{eqn:cramer1}, by the Cram\'{e}r's Theorem~\cite{dembo2009large} and the fact that $\mathcal{Y}^*(x)_{\mathcal{M}\setminus \mathcal{I}}$ is closed, one obtains that
  \begin{align*}
    &\liminf_{k\to\infty} \frac{1}{k}	\log \mathbb{P}_{x}\left(\avg({\bf Z}(k))_{\mathcal{M}\setminus \mathcal{I}} \in \mathcal{Y}^*(x)_{\mathcal{M}\setminus \mathcal{I}} \right) \\
    \geq & \inf_{z\in\R^m, z_{\mathcal{M}\setminus \mathcal{I}}\in\mathcal{Y}^*(x)_{\mathcal{M}\setminus \mathcal{I}}} \, d_x(z,\mathcal{M}\setminus\mathcal{I}) \\
    = & \inf_{z\in\mathcal{Y}^*(x)} \, d_x(z,\mathcal{M}\setminus\mathcal{I}) \\
    \geq & \inf_{z\in\mathcal{Y}^*(x)} \, d_x(z) \\
    \geq& u(\delta).
  \end{align*}
  Since the above argument holds for any $x$ and $\mathcal{I}$, one concludes that $r(f^*_{\delta},\delta) \geq u(\delta)$. Furthermore, $r(f^*_{\delta},\delta)$ is upper bounded by $u(\delta)$ due to Lemma~\ref{lemma:upperbound}, the proof is thus complete.
\end{proof}

\section{Proof of Lemma~\ref{lemma:vioassumpt2}}
\label{appendix:vioassumpt2}
Using the same argument as in the proof of Lemma~\ref{lemma:upperbound}, one readily obtains the second bullet of Lemma~\ref{lemma:vioassumpt2} from the first one. Therefore, we focus on the first bullet in the sequel.

The proof is of constructive nature. Since ${\bm H}$ is not $2q$-observable, without loss of generality, we let $H_{\mathcal{I}^*}$ is not of full column rank with $\mathcal{I}^*=\{2q+1,\ldots,m\}$. For any $\delta$,  let $x_1,x_2\in\R^n$ be any two vectors such that 
\begin{align*}
	H_{\mathcal{I}^*} (x_2-x_1) = 0,\: \text{and}\: \|x_2-x_1\| > \delta.
\end{align*} 
Let $\mathcal{I}_1=\{1,\ldots,q\}$ and $\mathcal{I}_2=\{q+1,\ldots,2q\}$. We then construct $y^*$ as follows:
\begin{align*}
	y^*_{\mathcal{I}^*} &= H_{\mathcal{I}^*} x_1, \\
	y^*_{\mathcal{I}_1} &= H_{\mathcal{I}_1} x_1, \\
	y^*_{\mathcal{I}_2} &= H_{\mathcal{I}_2} x_2.
\end{align*}
Then it is easy to verify that $d_{x_1}(y^*)=d_{x_2}(y^*)=0$. The proof is thus complete.

\section{Proofs of Lemmas in Section~\ref{sec:numericalImple}}
\label{proof:lemmas}
\begin{proof}[Proof of Lemma~\ref{lemma:restrictdistance}]
  For any index set $\mathcal{I}$, there holds
  \begin{align*}
    d_x(y,\mathcal{I}) =& \frac{1}{2} (y_{\mathcal{I}} - {\bm H}_{\mathcal{I}} x)^\top {\bm W}_{\{\mathcal{I}\}}^{-1} (y_{\mathcal{I}} - {\bm H}_{\mathcal{I}} x) \\
    =& \frac{1}{2} (\sqrt{{\bm W}_{\{\mathcal{I}\}}^{-1}}y_{\mathcal{I}} - \sqrt{{\bm W}_{\{\mathcal{I}\}}^{-1}}{\bm H}_{\mathcal{I}} x)^\top \\
                         &\: (\sqrt{{\bm W}_{\{\mathcal{I}\}}^{-1}}y_{\mathcal{I}} - \sqrt{{\bm W}_{\{\mathcal{I}\}}^{-1}}{\bm H}_{\mathcal{I}} x),
  \end{align*}
  which holds since ${\bm W}_{\{\mathcal{I}\}}^{-1}$ is a diagonal matrix. For simplicity of notation, in the remainder of this proof, we let $y_w= \sqrt{{\bm W}_{\{\mathcal{I}\}}^{-1}}y_{\mathcal{I}}$ and $H_w = \sqrt{{\bm W}_{\{\mathcal{I}\}}^{-1}}{\bm H}_{\mathcal{I}}$. By orthogonally projecting $y_w$ onto the range of $H_w$ using $H_wH_w^+$, where $H_w^+$ is the pseudo-inverse of $H_w$, one obtains
    \begin{align*}
  d_x(y,\mathcal{I}) =& \frac{1}{2} (y_w-H_wH_w^+y_w + H_wH_w^+y_w - H_w x)^\top \\
  &\: (y_w-H_wH_w^+y_w + H_wH_w^+y_w - H_w x). \addtag \label{eqn:pseduo}
  \end{align*} 
  Notice that $(y_w-H_wH_w^+y_w)$ is orthogonal to $(H_wH_w^+y_w - H_w x)$. 
  Furthermore, $W_i>0$ for each $i$, and by Assumption~\ref{assumpt:2qobservable}, ${\bm H}_{\mathcal{I}}$ is of full column rank for any $\mathcal{I}$ with $|\mathcal{I}|\geq m-2q$, $H_w$ is thus full column rank and $H_w^+=(H_w^\top H_w)^{-1}H_w^\top$. One then obtains~\eqref{eqn:distrewrite} from~\eqref{eqn:pseduo}. The proof is thus complete.
\end{proof}

\begin{proof}[Proof of Lemma~\ref{lemma:cvxopt}]
  It  follows readily from~\cite[Lemma 2.8]{yildirim2006minimum} that a ball $\mathcal{B}_\upsilon(c) \subseteq \R^n$ covers a full dimensional ellipsoid $\{x:(x-x_0)^\top Q (x-x_0)\leq \delta^2 \}$, where $Q\in\R^{n\times n}$ is positive definite and $\delta>0$, if and only if there exists $\tau\geq 0$ such that
 \begin{align}   \label{eqn:sos}
 \tau\begin{bmatrix}
 Q & -Qx_0 & 0 \\
 * & x_0^\top Q x_0 - \delta^2 &0 \\
 0 & 0 & 0
 \end{bmatrix}
 \succcurlyeq
 \begin{bmatrix}
 {\bm I}_n & -c & 0 \\
 -c^\top & -\upsilon^2 &c^\top \\
 0 & c & -{\bm I}_n
 \end{bmatrix}.
 \end{align}
 Also, it is clear that  a ball $\mathcal{B}_\upsilon(c) \subseteq \R^n$ covers a point $x\in\R^n$ if and only if 
 \[(x-c)^\top(x-c)\leq \upsilon^2,\]
 which, by Schur complement, is equivalent to 
 \begin{align} \label{eqn:sos1}
 \begin{bmatrix}
 \upsilon^2 & (x-c)^\top  \\
 * & {\bm I}_n
 \end{bmatrix}
 \succcurlyeq 0.
 \end{align}
 Furthermore, for any  $\phi$ such that $\mathfrak{I}(\phi)$ is not empty, the set $\mathcal{X}(y,\phi)$ is a union of some full dimensional ellipsoids (when the set $\mathfrak{I}_+(\phi)$ is not empty) and some single points (when the set $\mathfrak{I}_0(\phi)$ is not empty). Therefore, one can conclude Lemma~\ref{lemma:cvxopt}.
\end{proof}

\begin{proof}[Proof of Lemma~\ref{lemma:continuous}]
  It holds that $\mathcal{X}(y,\phi_0) \subseteq \mathcal{X}(y,\phi_1)$ for any $y$ and $\phi_0\leq \phi_1$. Therefore, $\rad(\mathcal{X}(y,\phi))$ is monotonically increasing w.r.t. $\phi$.

  Let $\res_{[i]}$ be the $i$-th item of the set $\{\res(\mathcal{I}): \mathcal{I}\subseteq \mathcal{M}, |\mathcal{I}| = m-q\}$ sorted in an ascending order. Then by viewing $\mathcal{X}(y,\phi)$ as a union of ellipsoids as in~\eqref{eqn:unionellipsoid}, one obtains that for any $\phi_0 \in \bigcup_{i=1}^{
    \begin{psmallmatrix}
      m\\q
    \end{psmallmatrix}-1
  }\left(\res_{[i]}, \res_{[i+1]}\right) \bigcup \left(\res_{[\begin{psmallmatrix}
        m\\q
      \end{psmallmatrix}]}, \infty\right)$, where $\begin{psmallmatrix}
    m\\q
  \end{psmallmatrix}$ is the binomial coefficient, the following holds:
  \begin{align*}
    \lim_{\phi\to\phi_0} \mathcal{X}(y,\phi) = \mathcal{X}(y,\phi_0),
  \end{align*}  
  and for any $\phi_0 \in \bigcup_{i=1}^{
    \begin{psmallmatrix}
      m\\q
    \end{psmallmatrix}
  }\res_{[i]}$,
  \begin{align*}
    \lim_{\phi\to\phi_0^+} \mathcal{X}(y,\phi) = \mathcal{X}(y,\phi_0)
  \end{align*}
  holds, where $\phi\to\phi_0^+$ means that $|\phi-\phi_0| \to 0$ and $\phi-\phi_0 >0$. Notice also that $\rad(\mathcal{X}(y,\phi)) =0$ for all $\phi\leq \res_{[1]}$. Therefore, one can conclude the first two bullets of Lemma~\ref{lemma:continuous}.
  
  By the first two bullets, in order to obtain the third one, it suffices to show that when $\phi$ is in any of the $
  	\begin{psmallmatrix}
  m\\q
  \end{psmallmatrix}$ intervals $\bigcup_{i=1}^{
  	\begin{psmallmatrix}
  	m\\q
  	\end{psmallmatrix}-1
  }\left(\res_{[i]}, \res_{[i+1]}\right) \bigcup \left(\res_{[\begin{psmallmatrix}
  	m\\q
  	\end{psmallmatrix}]}, \infty\right
  )$, $\rad(\mathcal X(y,\phi))$ is strictly increasing w.r.t. $\phi$. Notice that when $\phi$ is in any of these intervals,  $\mathfrak{I}_{0}(\phi)$ is empty and $\mathfrak{I}_{+}(\phi)$ remains the same, and, therefore, the optimal solution $\psi^*$ (i.e., the square of $\rad(\mathcal X(y,\phi))$) to the optimization problem in Lemma~\ref{lemma:cvxopt} is strictly increasing w.r.t. $\phi$. The third bullet is thus concluded and, therefore, the proof is complete.
\end{proof}

\section{Proof of Lemma~\ref{theorem:upperbound}}
\label{appendix:upperbound}

Let $x^*,s^*$ be the optimal solution to the optimization problem~\eqref{eqn:optupper}. Further let $x\in\R^n$ be any vector and $\mathcal{I}_0, \mathcal{I}_1$ the two index sets such that $\mathcal{I}_0\cup\mathcal{I}_1=\supp(s^*)$, $|\mathcal{I}_0|\leq q$, and $|\mathcal{I}_1|\leq q$. We then construct the following three quantities $x_0,x_1\in\R^n$ and $y^*\in\R^m$:
\begin{align}
	x_0 &= x-\delta x^*, \:\:
	x_1 =  x+\delta x^*,\label{eqn:con1}\\
	y^*_{\mathcal{M}\setminus\supp(s^*)} &= ({\bm H} x)_{\mathcal{M}\setminus\supp(s^*)},\label{eqn:con2}\\
  y^*_{\mathcal{I}_0} &= ({\bm H} x_0)_{\mathcal{I}_0},\:\:  y^*_{\mathcal{I}_1} = ({\bm H} x_1)_{\mathcal{I}_1}.  \label{eqn:con3}
\end{align}
Then one verifies that 
\begin{align}
	\|x_0-x_1\|_2 &= 2\delta, \label{eqn:distance} \\
	d_{x_0}(y^*) &\leq \bar{u}(\delta),\label{eqn:distance1}\\
	d_{x_1}(y^*) &\leq \bar{u}(\delta). \label{eqn:distance2}
\end{align}
Notice that~\eqref{eqn:distance1} holds because by the definition of $d_{x_0}(y^*)$ (i.e., the optimal value of~\eqref{eqn:opt1}), we have 
\begin{align*}
	d_{x_0}(y^*) &\leq \frac{1}{2}\sum_{i=1}^m (y^*_i - H_ix_0 +a_i)^2/W_i,\\
	&=\bar{u}(\delta)
\end{align*} 
where $a_i=-H_i(x_1-x_0)$ for $i\in\mathcal{I}_1$ and $a_i=0$ for $i\in\mathcal{M}\setminus\mathcal{I}_1$. The equation~\eqref{eqn:distance2} can be obtained in the same manner. 

Therefore, $x_0,x_1\in\mathcal{X}(y^*,\bar{u}(\delta))$ by~\eqref{eqn:distance1}~and~\eqref{eqn:distance2}. Combining~\eqref{eqn:distance}, one then obtains that 
\begin{align} \label{eqn:radiusgreater}
	\rad(\mathcal{X}(y^*,\bar{u}(\delta))) \geq \delta.
\end{align}
Notice that since $\bar{u}(\delta)=\delta^2\bar{u}(1)$, we have for any $\epsilon,\delta>0$, 
\[\bar{u}(\delta)+\epsilon = \bar{u}(\delta'),\]
where $\delta'=\delta\sqrt{(\bar{u}(\delta)+\epsilon)/\bar{u}(\delta)} > \delta$. Then using the same construction technique as in~\eqref{eqn:con1}-\eqref{eqn:con3}, one concludes that, by~\eqref{eqn:radiusgreater}, for any $\epsilon>0$, there exists $y\in\R^m$ such that
\begin{align*}
	\rad(\mathcal{X}(y,\bar{u}(\delta) + \epsilon)) &= \rad(\mathcal{X}(y,\bar{u}(\delta'))) \\
	&\geq \delta' >\delta.
\end{align*}
Therefore, from the definition of $u(\delta)$, one obtains that $u(\delta)\leq \bar{u}(\delta)$ for any $\delta$. The proof is thus complete.

\section{Proof of Theorem~\ref{theorem:uniform}}
\label{appendix:uniform}
The proof is divided into two parts. 

\textit{Part \Rmnum{1}.} We show that for every $y\in\R^m$, 
\begin{align} \label{eqn:sufficiency}
	d_x(y)\geq \bar{u}(|x-\trm(y)|)
\end{align}
holds for every $x\in\R$,
where recall that $\bar{u}(\delta)$ is the upper bound in Lemma~\ref{theorem:upperbound}.

Since the sensors are homogeneous, then without loss of generality, we let $W_i=W/2$ for any $1\leq i\leq m$. Then one obtains that 
\[\bar{u}(\delta) = (m-2q)\delta^2/W.\]
One further obtains that for any $y\in\R^m$ and $x\in\R$,
\begin{align*}
	d_x(y) &\geq \sum_{i=q+1}^{m-q} (y_{[i]}-x)^2/W \\
	&\geq (m-2q)(\frac{1}{m-2q}\sum_{i=q+1}^{m-q} y_{[i]} -x)^2/W \\
	&=\bar{u}(|x-\trm(y)|).
\end{align*}
Therefore,~\eqref{eqn:sufficiency} holds.

\textit{Part \Rmnum{2}.}  Notice that, for any $x,\delta$ and $k$,  if $y\in\mathcal{Y}(f^{\trm},x,\delta,k)$, $|x-\trm(y)|\geq \delta$ holds, where $\mathcal{Y}(\cdot,\cdot,\cdot,\cdot)$ is introduced in Definition~\ref{def:set}. Then \eqref{eqn:sufficiency} yields that  $d_x(y) \geq \bar{u}(\delta)$ for every $y\in\mathcal{Y}(f^{\trm},x,\delta,k)$. 

Using the same argument as in the proof of Lemma~\ref{lemma:achieve}, one obtains that for any $\delta$, 
\[r(f^{\trm},\delta) \geq \bar{u}(\delta) \geq u(\delta).\]
Due to the optimality of $u(\delta)$, the above equation holds as equality. The proof is thus complete.

    %
    %

\bibliographystyle{IEEEtran}
\bibliography{ref_xiaoqiang}

\end{document}